\RequirePackage{amsmath}
\documentclass{svjour3_arxiv}

\usepackage{graphicx}
\usepackage[hyphens]{url}  
\usepackage{amsmath,amsfonts}
\usepackage{mathtools}
\usepackage{xcolor}
\usepackage{tikz}
\usepackage{natbib}
\usepackage{hyperref}
\usepackage{lineno}

\usepackage{amssymb}
\usepackage{pifont}
\newcommand{\cmark}{\ding{51}}%
\newcommand{\xmark}{\ding{55}}%

\renewcommand{\d}{{\bf d}} 
\newcommand{\D}{{\bf D}} 
\newcommand{\DDD}{\mathcal{D}} 
\newcommand{\DD}{\mathfrak{D}} 
\newcommand{\C}{\mathcal{C}} 

\newcommand\bigutimes{\mathop{\ooalign{$\bigcup$\cr%
   \hfil\raise0.36ex\hbox{$\scriptscriptstyle\boldsymbol{\times}$}\hfil\cr}}}

\begin{document}

\title{United for Change: Deliberative Coalition Formation to Change the Status Quo\thanks{Earlier versions of this article have been accepted for presentation at the 35th AAAI Conference on Artificial Intelligence, AAAI-21~\citep{preliminary} and at the 8th International Workshop on Computational Social Choice, COMSOC-21. We wish to thank Schloss Dagstuhl---Leibniz Center for Informatics. The paper develops ideas 
discussed by the authors and other participants during the Dagstuhl Seminar 19381 (Application-Oriented Computational Social Choice), summer 2019. We would also like to thank the participants of the University of Bayreuth Colaform workshop and the Lorentz Center workshop on Modelling of Social Complexity in Argumentation for useful feedback on earlier versions of this work. We thank the generous support of the Braginsky Center for the Interface between Science and the Humanities.
Edith Elkind was supported by the ERC Starting Grant ACCORD (GA 639945).
Davide Grossi was supported by the \href{https://hybrid-intelligence-centre.nl}{Hybrid Intelligence Center}, a 10-year program funded by the Dutch
Ministry of Education, Culture and Science through the Netherlands
Organisation for Scientific Research (NWO).
Nimrod Talmon was supported by the Israel Science Foundation (ISF; Grant No. 630/19).}}


\author{
    Edith Elkind
    \and
    Davide Grossi
    \and
    Ehud Shapiro
    \and
    Nimrod Talmon
}


\institute{%
    Edith Elkind
    \at
    University of Oxford \\
    \email{elkind@cs.ox.ac.uk}
    \and
    Davide Grossi
    \at
    University of Groningen and University of Amsterdam \\
    \email{d.grossi@rug.nl}
    \and
    Ehud Shapiro
    \at
    Weizmann Institute of Science \\
    \email{ehud.shapiro@weizmann.ac.il}
    \and
    Nimrod Talmon
    \at
    Ben-Gurion University \\
    \email{talmonn@bgu.ac.il}
}

\date{}

\maketitle


\begin{abstract}
We study a setting in which a community wishes to identify a strongly supported proposal from a space of alternatives, in order to change the status quo.
We describe a deliberation process in which agents dynamically form coalitions 
around proposals that they prefer over the status quo. 
We formulate conditions on the space of proposals and on the ways in which coalitions are formed that guarantee deliberation to succeed, that is, to terminate by identifying a 
proposal with the largest possible support. 
Our results provide theoretical foundations for the analysis of deliberative processes such as the ones that take place
in online systems for democratic deliberation support.
%
\end{abstract}

\section{Introduction}



Social choice theory has provided us with a wealth of tools to design voting methods for democratic decision-making~\citep[see, e.g.,][]{comsoc-book}. However, besides voting, another important dimension of collective decision-making involves determining what should be put to vote in the first place. This aspect of democratic choice is referred to as the `right of initiative' in a parliamentary context, and 
has received considerably less attention in the computational social choice literature. 

Indeed, in practice, agents engaging in group decisions do not just vote 
on an externally determined set of choices: rather, they make proposals, 
deliberate over them, and join coalitions to push their proposals through.
Understanding deliberative processes is essential in order to provide a more comprehensive picture of democratic decision-making. This is important in the context of established democratic institutions, such as parliaments, but is 
especially critical for digital democracy applications~\citep{brill18interactive} 
that provide support for democratic decision-making in more informal, less structured communities.
Notable examples of this kind are the
LiquidFeedback\footnote{\url{https://liquidfeedback.org}} \citep{liquid_feedback} and the Polis\footnote{\url{https://pol.is/home}} platforms. 

In this paper, we aim 
to provide a parsimonious model of deliberative processes in self-governed (online) systems. 
We abstract away from the communication mechanisms by means of which deliberation may be 
concretely implemented, and focus instead on the coalitional effects of deliberation, that is, 
on how coalitions in support of various proposals 
may be formed or broken in the face of new suggestions. 
In our model, agents and alternatives are located in a metric space, 
and there is one distinguished alternative, which we refer to as the {\em status quo}. 
We assume that the number of alternatives 
is large (possibly infinite), so that the agents cannot be expected to rank the alternatives
or even list all alternatives they prefer to the status quo. Rather, during the deliberative 
process, some agents formulate new proposals, and then each agent can decide
whether she prefers a given proposal~$p$ to the status quo, i.e., whether $p$
is closer to her location than the status quo alternative; if so, she may join a coalition of
agents supporting $p$. This process is democratic in that each participant who is capable of formulating a new proposal is welcome to do so; furthermore, participants who do not have the 
time or sophistication to work out a proposal can still take part in 
the deliberation by choosing which coalition to join.

These coalitions are dynamic and change over time: an agent may move to another
coalition, or two coalitions may merge, possibly leaving some members behind. 
We assume that agents are (1) {\em consensus-seeking}, i.e., they aim to 
jointly identify a proposal that has large support, and (2) {\em myopic}, 
i.e., they make their moves without trying to predict their impact 
on other agents' behavior in the subsequent steps. As a consequence, 
in our model agents always aim to increase the size of the coalition that they are
part of, as long as this coalition supports a proposal they prefer to the status quo. 
In particular, they may move from a smaller coalition whose proposal
is close to their ideal point to a larger coalition whose proposal is further
away from their ideal point (as long as they still prefer the latter coalition's proposal
to the status quo).
We consider a deliberation to be successful if it identifies 
a most popular alternative to the status quo. 
Thus, our results should be interpreted as an attempt 
to map out what stands in the way of successful deliberation 
even when assuming that the agents' primary motivation is to reach a consensus. 

 
To flesh out the broad outline of the deliberative coalition formation process described above, 
one needs to specify rules that govern the dynamics of coalition formation. 
We explore several such rules, ranging from single-agent moves 
to more complex transitions where coalitions merge behind a new proposal, possibly
leaving some dissenting members behind. Our aim is to understand whether 
the agents can succeed at identifying
credible alternatives to the status quo if they conduct deliberation in a certain way.

We concentrate on how the properties of the underlying abstract space of proposals and the 
coalition formation operators available to the agents affect 
the success of the deliberative coalition formation process.
We show that, as the complexity of the proposal space increases, more sophisticated forms of coalition 
formation are required in order to assure success. Intuitively, this seems to suggest that 
complex deliberative spaces require more sophisticated coalition formation abilities on the side of the agents. 

\paragraph{Paper contribution}
We study five ways in which agents can form coalitions:
\begin{enumerate}

\item
by {\em deviation}, when a single agent from one coalition moves to another, weakly larger, coalition;  

\item
by {\em following}, when a coalition joins another 
coalition in supporting the second coalition's proposal; 

\item
by {\em merge}, when two coalitions join forces behind a new proposal; 

\item
by {\em compromise}, when agents belonging to two coalitions form a new larger coalition, possibly 
leaving dissenting agents behind;

\item 
by {\em multi-party compromise}, when agents belonging to possibly more than two coalitions form a new larger coalition, perhaps leaving dissenting agents behind.
\end{enumerate}
We refer to these types of coalition formation operations as {\em transitions}. Even though these transitions are all inspired by intuitive features of coalition formation in deliberation, we emphasise that we are not making empirical claims about the frequency or popularity of these transitions in real-world deliberations. Rather, they should be treated as abstract constructs that enable us to describe a hierarchy of coalition formation modes of increasing sophistication; this, in turn, puts us in a position to make claims to the effect that certain proposal spaces require more elaborate coalition formation procedures.

We show that, for each class of transitions, the deliberation process is guaranteed
to terminate; in fact, for single-agent, follow, and merge transitions
the number of steps until convergence is at most polynomial in the number of agents.
Furthermore, follow transitions are sufficient
for deliberation to succeed if the set of possible proposals is a subset of the $1$-dimensional
Euclidean space, but this is no longer the case in two or more dimensions; in contrast, 
compromise transitions are enough in sufficiently rich subsets of ${\mathbb R}^d$
for each $d\ge 1$. The `richness' condition, however, is essential:
  we provide an example where the space of proposals is a finite subset of ${\mathbb R}^3$, 
but compromise transitions are not capable of identifying a most popular proposal.

While our primary focus is on deliberation in Euclidean spaces, we also provide results for two classes of non-Euclidean spaces. The first one is the class of weighted trees, 
with the metric given by the standard path length. We show that, in these spaces, merge transitions suffice for successful deliberation. The second class of spaces arises naturally in combinatorial domains as well as in multiple referenda settings. These spaces, which we refer to as $d$-hypercubes, are defined by the set of all binary opinions on $d$ issues with the discrete metric known as the Hamming, or Manhattan, distance. We show that deliberative success in $d$-hypercubes requires multi-party compromises involving at least $d$ coalitions, and possibly more. In particular, we prove that $d$-compromises guarantee success if  $d\le 3$, but when $d = 4$ we may need $5$-compromises.
Table~\ref{table:results} provides an overview of our findings, with pointers to the relevant results in the paper.

\begin{table}[t]
	\centering
	\resizebox{\linewidth}{!}{
		\begin{tabular}{l l l l l l l} 
		    Transition
			& Termination
			& $1$-Euclidean
			& $2$-Euclidean
			& $d$-Euclidean
			& Trees
			& $d$-Hypercubes
			\\ \hline
			Single-ag.
			& $ n^2$~(Pr.~\ref{prop:single-terminates})
			& \xmark~(Ex.~\ref{ex:single-bad})
			& \xmark
			& \xmark
			& \xmark
			& \xmark
			\\
			Follow
			& $ n^2$~(Pr.~\ref{prop:follow-terminates})
			& \cmark~(Th.~\ref{thm:follow-success-1D})
			& \xmark~(Ex.~\ref{ex:2d})
			& \xmark
			& \xmark
			& \xmark
			\\
			Merge
			& $ n^2$~(Pr.~\ref{prop:merge-terminates})
			& \cmark~(Th.~\ref{thm:follow-success-1D})
			& \xmark~(Ex.~\ref{ex:needcomp})
			& \xmark
			& \cmark~(Th.~\ref{thm:merge-success-tree})
			& \xmark
			\\
			Compr.
			& $n^n$~(Pr.~\ref{prop:lex})
			& \cmark
			& \cmark
			& \cmark~(Th.~\ref{thm:success})
			& \cmark
			& \xmark~(Ex.~\ref{ex:3hyper}) 
			\\
			$d$-Compr.
			& $n^n$~(Co.~\ref{cor:lex})
			& \cmark
			& \cmark
			& \cmark
			& \cmark
			& \cmark~($d \leq 3$, Pr.~\ref{prop:hdone}, Pr.~\ref{prop:34hyper})
	\end{tabular}
	}
	\caption{Summary of our main results. For each transition type, the respective row shows bounds on the length of deliberation, and whether success is guaranteed for various metric spaces.}
	\label{table:results}
\end{table}


We view our work as an important step towards modeling a form of pre-vote deliberation 
usually not studied within the social choice literature: how voters can identify 
proposals with large support. We believe that our formal model provides the basis 
for a framework within which systems designed to support democratic deliberation 
could be modeled and analyzed. In particular, it captures a highly desirable common feature 
of such systems, including, e.g., LiquidFeedback and Polis: 
enabling participants to come up with novel proposals and to submit them for the approval 
(or disapproval) of other participants. Only proposals commanding large-enough support would then 
be considered as options to be voted upon in a formal voting phase. 
Hence, participants are nudged towards identifying positions that are mutually beneficial 
(see also the work of~\cite{pietro16system}) and can therefore elicit the approval of larger groups. 
Our theory aims to make the first steps towards providing analytical foundations for such systems.


\paragraph{Structure of the paper} 
After discussing related work (Section~\ref{section:relatedwork}), we describe our formal model in Section~\ref{section:process}.
Then, we concentrate on Euclidean deliberation spaces and study the power of different deliberation operations:
  specifically,
  in Section~\ref{section:singleagent}, we consider single-agent transitions;
  in Section~\ref{section:follow} we consider follow transitions;
  in Section~\ref{section:merge} we consider merge transitions 
     (also in the context of weighted trees);
  in Section~\ref{section:compromise} we consider compromise transitions.
Finally, in Section~\ref{section:hypercube}, we study deliberation in hypercubes and multi-party compromise transitions. We conclude in Section~\ref{sec:concl}.


\section{Related Work}\label{section:relatedwork}

\paragraph{Deliberation.} Group deliberation has been an object of research in several disciplines, 
from economics to political theory and artificial intelligence. 
A wealth of different approaches to deliberation can be identified, reflecting the complexity of the concept.
\citet{list11group} has focused on axiomatic approaches to deliberation, viewed as an opinion transformation function. Experimental work---relying on either lab experiments \citep{list13deliberation} or simulations \citep{roy20deliberation}---has then tried to assess the effects of deliberations on individual opinions (e.g., whether deliberation facilitates the set of opinions becoming single-peaked).
Deliberation has also been approached from a game-theoretic perspective \citep{landa09game}. With the tools of game theory, deliberation has been studied from a number of angles: as a process involving the exchange of arguments for or against positions \citep{chung20formal,patty08arguments}; as a process of persuasion \citep{glazer01debates,glazer04optimal,glazer06study}; as the pre-processing of inputs for voting mechanisms \citep{austen-smith05deliberation,perote15model,karanikolas2019voting}; as a mechanism enabling preference discovery \citep{hafer07deliberation}. Deliberation has also been studied as a distributed protocol for large-scale decision-making involving only sequential local interaction in small groups, with a focus on algorithmic complexity \citep{goel2016towards,fain2017sequential}.

Our work sets itself somewhat apart from the above literature by focusing exclusively on the consensus-seeking aspect of deliberation. Deliberation is studied here as a distributed process whereby widely supported proposals can be identified in a decentralized way. In doing so, we abstract away from: the concrete interaction mechanisms (e.g., argumentation) that agents may use to communicate and assess
proposals; strategic issues that may arise during communication; as well as how deliberation might interact with specific voting rules. Instead, we concentrate on the results of such interactions, as manifested by changes in the 
structure of coalitions supporting different proposals. 
This focus also sets our work apart from  influential opinion dynamics models (e.g., \citep{degroot74reaching}) where deliberation is driven by social influence rather than by coalition formation via, for instance, compromise.

Importantly, we abstract away from strategic issues that agents may have to  confront when deciding 
to join or leave coalitions; this differentiates our work from the related literature on dynamic 
coalition formation \citep{arnold02dynamic,chalkiadakis08sequential}.
In particular, \citet{chalkiadakis08sequential} concentrate on uncertainties that the agents may experience; and \citet{arnold02dynamic} focus on agents that receive payoffs, which depend on the coalitions they are in.

\paragraph{Spatial voting.} From a technical point of view, our contribution is closely related to work on spatial voting that stems from the Hotelling model \citep{hotelling29stability}. Many social choice settings are naturally embedded in a metric space,
and there is a large literature that considers preference aggregation
and coalition formation in such scenarios \citep{coombs,MG99,devries99governing}.
In this context we mention, in particular, the work of \citet{metricspaces}, which considers a framework for voting and proposing, instantiated to several metric spaces. It includes an agent population where each agent is associated with an ideal point. This framework models a broad range of social choice settings, and is intended to be used for deliberative decision-making. 
\citet{metricspaces} evaluate
several general aggregation methods for their framework 
based on their computational and normative properties. Our work extends the framework of \citet{metricspaces} by accommodating the processes of coalition formation.


\paragraph{Coalition formation.} 
Since we model deliberation as the evolution of coalition structures in response to 
deviations by groups of agents, our framework is closely related to that of hedonic games 
\citep{aziz16handbook}. Hedonic games offer a rich framework for studying 
stability properties of coalition structures where agents have preferences 
over coalitions that they can be part of; here, stability is 
interpreted as the absence of profitable deviations, by individual agents or groups of 
agents. In a similar spirit, we are also concerned with stability of coalition structures 
that arise in a chain of deviations. Indeed, some of the transitions types that we investigate
are inspired by stability concepts developed in the hedonic games literature
(e.g., single-agent transitions and $\ell$-compromise transitions are closely related
to the deviations underpinning the concepts of, respectively, 
Nash stability and core stability in hedonic games). 
However, a key difference is that in our framework each 
coalition is associated with a specific proposal, and the choice of proposal is an important 
factor in determining which agents would like to join/abandon this coalition. In this regard, 
our setting is closer to that of group activity selection \citep{gasp}, where agents
form groups to engage in activities. However, to the best 
of our knowledge, there is no existing work on group activity selection with metric 
preferences, and some of the transition types we consider (such as merge, follow and 
subsume transitions) have not received
much attention in the context of group activity selection. 


\section{Formal Model}\label{section:process}
We view deliberation as a process in which agents aim to find an alternative 
preferred over the status quo by as many agents as possible.
Thus, we assume a (possibly infinite) domain $X$ of {\em alternatives}, or {\em proposals}, 
which includes the {\em status quo}, or \emph{reality}, $r \in X$. We also assume 
a set $V = \{v_1, \ldots, v_n\}$ of $n$ {\em agents}. 
For each proposal $x\in X$, an agent $v$ is able to articulate whether 
she (strictly) prefers $x$ over the status quo $r$ (denoted as $x >_v r$); when $x >_v r$ we say that $v$ {\em approves}~$x$. For each $v\in V$, 
let $X^v=\{x\in X: x>_v r\}$; the set $X^v$ is the \emph{approval set} of $v$. 
Given a subset of agents $C \subseteq V$ and a proposal $p\in X$, 
let furthermore $C^p := \{v \in C : p >_v r\}$. Then, the agents in $C^p$ are the \emph{approvers} of $p$ in $C$.

Throughout this paper, in order to represent agents' views on a possibly infinite set of proposals, we focus on the setting where $X$ and $V$ are contained in a metric space $(M, \rho)$, i.e., (1) $X, V\subseteq M$,
(2) the mapping $\rho: M\times M\to{\mathbb R}^+\cup\{0\}$ satisfies (i) $\rho(x, y)=0$
if and only if $x=y$, (ii) $\rho(x, y)=\rho(y, x)$, and (iii) $\rho(x, y)+\rho(y, z)\le \rho(x, z)$
for all $x, y, z\in M$, and
(3) for every $x\in X$ and every $v\in V$ we have $x>_v r$ if and only if $\rho(v, x)<\rho(v, r)$. E.g., if $M = {\mathbb R}^2$
and the metric is the usual Euclidean metric in ${\mathbb R}^2$,
then the approval set of $v$ consists of all points in $X$ that are 
located inside the circle with center $v$
and radius $\rho(v, r)$ (see Figure~\ref{fig:ex1}), 
whereas the set of supporters of a proposal $p$ in $C$
consists of all agents $v\in C$ such that $\rho(v, p) < \rho(v, r)$.
Geometrically, consider the perpendicular bisector to $[p, r]$, i.e., 
the line $\ell$ that passes through the midpoint of the segment $[p, r]$
and is orthogonal to it. 
Then, $C^p$ consists of all agents $v\in C$ such that $v$ and $p$
lie on the same side of $\ell$.

Thus, an instance of our problem can be succinctly encoded by a $4$-tuple $(X, V, r, \rho)$;
we will refer to such tuples as {\em deliberation spaces}.

\begin{remark}
Note that we do not require that $X=M$. By allowing $X$ to be a proper subset of $M$ 
we can capture the case where the space of proposals is, e.g., a finite subset of ${\mathbb R}^d$
for some $d\ge 1$. Moreover, in our model it need not be the case that $V\subseteq X$, 
i.e., we do not assume that for each agent there exists a `perfect' proposal.
Furthermore, while for each agent $v$ the quantities $\rho(v, x)$ are well-defined
for each $x\in X$, we do not expect the agents to compare different proposals
based on distance; rather, the distance only determines which proposals
are viewed as acceptable (i.e., preferred to the status quo).
\end{remark}

Agents proceed by forming coalitions around proposals.
Thus, at each point in the deliberation, agents can be partitioned into coalitions, 
so that each coalition $C$ is identified with a proposal $p_C$ and all agents in $C$
support $p_C$. Agents may then move from one coalition to another as well as 
select a proposal from $X$ that is not associated with any of the existing 
coalitions and form a new coalition around it. We consider a variety 
of permissible moves, ranging from single-agent transitions 
(when an agent abandons her current coalition and joins a new one), 
to complex transitions that may involve agents from multiple 
coalitions and a new proposal.
In each case, we assume that agent~$v$ is unwilling
to join a coalition if this coalition advocates a proposal $p$ such that $v$
(weakly) prefers the status quo $r$ to $p$.

We now introduce the two main objects to be studied in our work: deliberative coalitions and deliberative coalition structures.

\begin{definition}[Deliberative Coalition]\label{def:coalition}
A \emph{deliberative coalition} is a pair $\d=(C, p)$, where
$C\subseteq V$ is a set of agents, $p\in X$ 
is a proposal, and either 
(i)  $p=r$ and $x\not>_v r$ for all $v\in C, x\in X\setminus\{r\}$, or
(ii) $p\neq r$ and $p >_v r$ for all $v\in C$.
We refer to $p$ as the \emph{supported proposal} of~$\d$. 
The set of all deliberative coalitions is denoted by $\DDD$.
\end{definition}
When convenient, we identify a coalition 
$\d = (C,p)$ with its set of agents $C$ 
and write $\d^q := C^q$ for $q\in X$, 
$|\d| := |C|$. 

\begin{remark}
While we allow coalitions that support the status quo,
we require that such coalitions consist of agents who weakly prefer the status quo 
to all other proposals in $X$. We discuss a relaxation
of this constraint in Section~\ref{sec:concl}.
\end{remark}

A partition of the agents into deliberative coalitions 
is called a deliberative coalition structure.

\begin{definition}[Deliberative Coalition Structure]
\label{def:structure}
A {\em deliberative coalition structure} ({\em coalition structure} for short) 
is a set $\D = \{\d_1, \ldots, \d_m\}$, $m \ge 1$, such that: 
\begin{itemize}
    \item $\d_i = (C_i,p_i)\in\DDD$ for each $i \in [m]$;
    \item $\bigcup_{i\in [m]}C_i = V$, $C_i \cap C_j = \emptyset$ for all $i,j \in [m]$ with $i\neq j$.
\end{itemize}
The set of all deliberative coalition structures over $V$ and $X$ is denoted by $\DD$.
\end{definition}

Note that a deliberative coalition structure may contain several coalitions 
supporting the same proposal; also, for technical reasons we allow empty 
deliberative coalitions, i.e., coalitions $(C, p)$ with $C=\emptyset$.

\begin{example}
Consider a set of agents $V = \{v_1, v_2, v_3\}$ and a space of proposals $X = \{a, b, c, d, r\}$,
where $r$ is the status quo. 
Suppose that $X^{v_1} = \{a, b\}$, $X^{v_2} = \{b, c\}$, and $X^{v_3} = \{b, c, d\}$. 
Then for $C = \{v_1, v_2\}$ we have 
$C^a = \{v_1\}$ and $C^b = \{v_1, v_2\}$. 
Furthermore, let $C_1=\{v_1, v_2\}$, $C_2=\{v_3\}$, and let 
$\d_1 = (C_1, b)$, 
$\d_2=(C_2, c)$.
Then $\D = \{\d_1, \d_2\}$ is a deliberative coalition structure;
see Figure~\ref{fig:ex1} for an illustration.

\begin{figure}
	\centering

	\begin{tikzpicture}[scale=1.3]
	
	\node (1) at (0,0) [label=above:$\textcolor{red}{v_1}$] {};
	\fill[red] (1) circle [radius=0.05];
	\draw (0.05,0) circle (1.5cm);
	
	\node (2) at (1,1) [label=above:$\textcolor{red}{v_2}$] {};
	\fill[red] (2) circle [radius=0.05];
	\draw (1,1) circle (1.5cm);
	
	\node (3) at (2,1) [label=above:$\textcolor{red}{v_3}$] {};
	\fill[red] (3) circle [radius=0.05];
	\draw (2,1) circle (1.5cm);
	
	\node (r) at (1.49,-0.4) [label=below:$r$] {};
	\fill (r) circle [radius=0.05];
	
	\draw[dotted] (1)--(r);
	\draw[dotted] (2)--(r);
	\draw[dotted] (3)--(r);
	
	\node (a) at (0,-1) [label=above:$\textcolor{blue}{a}$] {};
	\fill[blue] (a) circle [radius=0.05];
	\node (b) at (1.20,0.45) [label=left:$\textcolor{blue}{b}$] {};
	\fill[blue] (b) circle [radius=0.05];
	\node (c) at (1.5,1.8) [label=above:$\textcolor{blue}{c}$] {};
	\fill[blue] (c) circle [radius=0.05];
	\node (d) at (3,0.5) [label=above:$\textcolor{blue}{d}$] {};
	\fill[blue] (d) circle [radius=0.05];
	
	\end{tikzpicture}
	
	\caption{A deliberation space with $X = \{a, b, c, d, r\}$. The circle with center $v_i$, $i \in [3]$, contains all proposals approved by $v_i$.}
	\label{fig:ex1}
\end{figure}
\end{example}


\subsection{Deliberative Transition Systems}
As suggested above, we model deliberation as a process whereby deliberative coalition structures change 
as their constituent coalitions evolve.
We provide general definitions for deliberative processes, 
modeled as transition systems (Definitions \ref{definition:DTS}--\ref{definition:success}).
Subsequently, we explore several specific kinds of transitions.
Generally, a \emph{transition system} is characterized by a set of \emph{states} $S$, 
a subset $S_0 \subseteq S$ of \emph{initial states}, and a set of {\em transitions} $T$, 
where each transition $t \in T$ is represented by a pair of states $(s, s') \in S\times S$; 
we write $t=(s, s')$ and $s \xrightarrow{t} s'$ interchangeably. 
We use $s \xrightarrow{T} s'$ to denote that $s \xrightarrow{t} s'$ for some $t \in T$.
A \emph{run} of a transition system is a (finite or infinite) sequence  
$s_0 \xrightarrow{T} s_1 \xrightarrow{T} s_2 \cdots$; such a run is \emph{initialized} if $s_0 \in S_0$.
The last state of a finite run is called its {\em terminal state}.

\begin{definition}[Deliberative Transition System]\label{definition:DTS}
A \emph{deliberative transition system} over a set of proposals $X$
and a set of agents $V$ is a transition system 
that has $\DD$ as its set of states, a subset of states $\DD_0 \subseteq \DD$ 
as its set of initial states, 
and a set of transitions~$\mathcal{S}$. 
\end{definition}


Deliberative transition systems describe all possible sequences of deliberative coalition structures that can be obtained by executing a given set of transitions, starting from a given set of initial deliberative coalition structures. As such, they can be used to analyze the dynamics of the deliberation process. 
We define deliberations as maximal runs of such systems.

\begin{definition}[Deliberation]
\label{def:deliberation}
A \emph{deliberation} is a maximal run of a deliberative transition system, 
that is, a run that does not occur as a prefix 
of any other run (i.e., cannot be extended). 
\end{definition}


A \emph{successful} deliberation is one that identifies some of the most popular proposals
in $X\setminus\{r\}$.
In particular, if there is a majority-approved proposal, 
a successful deliberation process allows the agent population to identify some such proposal.
Observe, however, that it may be the case that no alternative is approved
by a majority of voters.

\begin{definition}[Most-Approved Alternatives]
\label{def:most}
For a set of agents $V$ and a set of proposals $X$, 
the set of \emph{most-approved alternatives} is
\linebreak
$M^* = \text{argmax}_{x \in X\setminus\{r\}} |V^x|$ 
and the \emph{maximum approval} is $m^* = \text{max}_{x \in X\setminus\{r\}} |V^x|$.
%
\end{definition}

Note that, as long as all approval sets are non-empty, $M^* \neq \emptyset$ 
and $|V^x| = m^*$ for every $x \in M^*$.
We are now ready to define success in deliberation.

\begin{definition}\label{definition:success}
A coalition $\d = (C,p)$ is \emph{successful} if 
$p\in M^*$
and 
$|C| = m^*$.
A deliberative coalition structure is \emph{successful} if it contains a successful coalition, 
and {\em unsuccessful} otherwise. A deliberation is \emph{successful} if it is finite 
and its terminal coalition structure is successful. 
\end{definition}
Intuitively, a coalition is successful if it consists of all approvers of a most-approved proposal. 
Successful deliberations are maximal finite runs that lead to a successful coalition structure, i.e., 
to a coalition structure that contains some successful coalition.

\begin{example}
Consider a toy example of a deliberative transition system that consists of 
(1) a single initial coalition structure $\D_0 = \{\d_1, \d_2, \d_3\}$,
with 
$\d_1 = (C_1, a)$ where $C_1 = \{v_1, v_2\}$,
$\d_2 = (C_2, b)$ where $C_2 = \{v_3, v_4\}$,
and
$\d_3 = (C_3, c)$ where $C_3 = \{v_5, v_6\}$; and 
(2) the set of transitions $\mathcal S$ only containing a transition $t$ from 
$\D_0$ to $\D_1 = \{\d_4, \d_5\}$,
with 
$\d_4 = (C_4, e)$ where $C_4 = \{v_1, v_2, v_3, v_4, v_5\}$,
$\d_5 = (C_5, f)$ where $C_5 = \{v_6\}$.
Then, assuming that there is no proposal approved by all agents, 
we have that $\D_0 \xrightarrow{t} \D_1$ is a successful deliberation.
\end{example}

\begin{remark}
Since successful deliberations identify proposals that receive the maximum number of 
approvals, a successful deliberative coalition structure identifies a proposal that would be 
a (possibly tied) winner of an election if agents were to submit approval ballots over the 
entire domain $X$ (or, more precisely, $X\setminus\{r\}$---
our model implicitly assumes that no agent approves $r$). 
However, there are two crucial differences between our deliberation model 
and conducting an approval-based election over $X$:
\begin{itemize}

\item First, our model of deliberation is decentralized in the sense that, at every step, agents do not need to report their approval/disapproval of each proposal, as they would have to do under a centralized voting mechanism. 

\item Second, our model is parsimonious in the sense that, at any time during the deliberation, each agent supports just a single proposal in her (possibly infinite) approval set. Thus, the agents do not need to be aware of the entire (possibly infinite) space of proposals $X$, whereas under a voting mechanism they
would have to form an opinion on each proposal in $X$.

\end{itemize}
Our aim is to formulate conditions guaranteeing that such a deliberative process succeeds, i.e., identifies proposals that could arise as outcomes of an approval election on the entire space $X$.
\end{remark}

Importantly, groups of agents may differ with respect to the types of agreements they are able to negotiate. To understand what can be accomplished by groups 
with particular sets of negotiation skills, we consider specific types of transitions; each can be thought of as a \emph{deliberation operator} that might be available to the agents.
For each transition type, we aim to determine if a deliberation
that only uses such transitions is guaranteed to terminate, and, if so,
whether the final coalition structure is guaranteed to be successful. We show that
the answer to this question depends on the underlying metric space:
simple transition rules guarantee success in simple metric spaces, but
may fail in richer spaces.

\section{Single-Agent Transitions}\label{section:singleagent}

The simplest kind of transition we consider is a deviation by a single agent.
As we assume that agents aim to form a successful coalition and are not necessarily 
able to distinguish among approved proposals, it is natural to focus on transitions where
an agent moves from a smaller group to a larger group; of course, this move is only possible
if the agent approves the proposal supported by the larger group.
%

\begin{definition}[Single-Agent Transition] \label{def:single}
A pair $(\D, \D')$ of coalition structures forms a {\em single-agent transition}
if there exist coalitions $\d_1, \d_2\in \D$ and $\d'_1, \d'_2\in \D'$
such that $|\d_2|\ge |\d_1|$, $\D\setminus\{\d_1, \d_2\}=\D'\setminus\{\d'_1, \d'_2\}$, 
and there exists an agent $v\in \d_1$ such that 
$\d'_1=\d_1\setminus\{v\}$, and $\d'_2=\d_2\cup\{v\}$.\footnote{Note that $\d'_1$ 
may be empty; we allow such `trivial' coalitions as it simplifies our definitions.}
We refer to $v$ as the \emph{deviating agent}. 
\end{definition}

Since $\D'$ is a deliberative coalition structure, agent $v$ must
approve the proposal supported by $\d_2$. As a consequence, no agent can deviate
from a coalition that supports $r$ to a coalition that supports some $p\in X\setminus\{r\}$
or vice versa.

Next, we show that a sequence of single-agent transitions necessarily
terminates after polynomially many steps.\footnote{We provide several such results throughout the paper. They establish worst-case guarantees on the length of a deliberation expressed in terms of the number of coalition formation operations. Intuitively, and in line with practice in algorithm analysis, we prefer deliberations to converge after a number of steps that is bounded by a polynomial expression in $n$ in the worst case. Similarly, we would view deliberations that perform a number of steps that is exponential in $n$ to be impractical (see, for instance, Proposition \ref{prop:lex}).}
\begin{proposition}\label{prop:single-terminates}
A deliberation that consists of single-agent transitions 
can have at most $n^2$ transitions.
\end{proposition}
\begin{proof}
Given a coalition structure $\D=\{\d_1, \dots, \d_m\}$ such that $\d_i=(C_i, p_i)$
for each $i\in [m]$, let 
\begin{align}
\lambda(\D) & =\sum_{i\in [m]}|C_i|^2. \label{eq:potential}
\end{align}
We will refer to $\lambda(\D)$ as the {\em potential} of $\D$.
Consider a single-agent transition where an agent moves from a coalition
of size $x$ to a coalition of size $y$; note that $1\le x\le y$. This move changes
the potential by $(y+1)^2+(x-1)^2-y^2-x^2=2 + 2(y-x)\ge 2$. 

We claim that for every 
deliberative coalition structure $\D$ over $n$ agents we have $\lambda(\D)\le n^2$.
Indeed, for the coalition structure $\D_0$ where all agents are in one coalition
we have $\lambda(\D_0)=n^2$. 
On the other hand, if a coalition structure contains
two non-empty coalitions $\d_1$, $\d_2$ with $|\d_1|\le |\d_2|$, then the calculation 
above shows that we can increase the potential by moving one 
agent from $\d_1$ to $\d_2$.  Further, every coalition structure can be transformed into $\D_0$ by a sequence of such moves, and hence  $\lambda(\D)\le \lambda(\D_0)\le n^2$ for each $\D\in\DD$.

As every single-agent transition increases the potential by at least $2$, 
and the potential takes values in $\{1, \dots, n^2\}$, the result follows.
\qed\end{proof}

However, a deliberation consisting of single-agent transitions is not necessarily successful, even
for very simple metric spaces. The next example shows that such a deliberation
may fail to identify a majority-approved proposal even if the associated metric space is
the 1D Euclidean space.

\begin{figure}
	\centering
		\begin{tikzpicture}[scale=1.3]
	
	\node (r) at (0,0) [label=below:$r$] {};
	\fill (r) circle [radius=0.05];
	
	\node (a) at (1,0) [label=below:$\textcolor{blue}{a}$,label=above:$\textcolor{red}{v_1,v_2,v_3}$] {};
	\fill[blue] (a) circle [radius=0.05];
	
	\node (b) at (5,0) [label=below:$\textcolor{blue}{b}$,label=above:$\textcolor{red}{v_4,v_5,v_6,v_7}$] {};
	\fill[blue] (b) circle [radius=0.05];
	
	\node (c) at (-1,0) [label=below:$\textcolor{blue}{c}$,label=above:$\textcolor{red}{v_8,v_9,v_{10}}$] {};
	\fill[blue] (c) circle [radius=0.05];
	
	\node (E) at (6,0) {};
	\node (W) at (-2,0) {};	
	
	\draw[->] (r) -- (E);
	\draw[->] (r) -- (W);
	
	\end{tikzpicture}
	\caption{The metric space in Example~\ref{ex:single-bad}.}
	\label{fig:1D}
\end{figure}

\begin{example}\label{ex:single-bad}
Suppose that $X, V\subseteq {\mathbb R}$, and 
$X=\{r, a, b, c\}$ with
$r=0$, $a=1$, $b=5$, $c=-1$. There are three agents $v_1, v_2, v_3$ located at $a$, 
four agents $v_4, v_5, v_6, v_7$ located at $b$, 
and three agents $v_8, v_9, v_{10}$ located at $c$
(see Figure \ref{fig:1D}).
Observe that a majority of the agents prefer $a$ to $r$.
Consider the deliberative coalition structure $\{\d_1, \d_2, \d_3\}$ 
with
$\d_1=(\{v_1, v_2, v_3\}, a)$,
$\d_2=(\{v_4, v_5, v_6, v_7\}, b)$,
$\d_3=(\{v_8, v_9, v_{10}\}, c)$. 
There are no single-agent transitions from this 
coalition structure: in particular, the agents in $\d_2$ do not want
to deviate to $\d_1$ because $|\d_1|<|\d_2|$, and the agents in $\d_1$
do not want to deviate to $\d_2$, because they do not approve $b$.
Note that this argument still applies to any proposal space $X'$ with $X\subseteq X'$;
e.g., we can take $X'=\mathbb R$.
\end{example}

Thus, to ensure success, we need to consider more powerful transitions.

\section{Follow Transitions}\label{section:follow}


Instead of considering moves by a single agent, 
we will now focus on moves by entire coalitions; specifically, 
we consider transitions where all members of a coalition join 
another coalition in supporting 
that coalition's current proposal.

\begin{definition}[Follow Transition]
A pair of coalition structures $(\D, \D')$ forms a {\em follow transition}
if there exist non-empty coalitions $\d_1, \d_2\in \D$ and $\d'_2\in \D'$
such that $\d_1=(C_1, p_1)$, $\d_2=(C_2, p_2)$,
$\D\setminus\{\d_1, \d_2\}=\D'\setminus\{\d'_2\}$, 
and $\d'_2 = (C_1\cup C_2, p_2)$. 
\end{definition}
Note that each follow transition reduces the number of coalitions by one, 
so a deliberation that consists of follow transitions only cannot take more than 
$n-1$ steps. Also, $|\d'_2|^2 = (|\d_1|+|\d_2|)^2>|\d_1|^2+|\d_2|^2$, i.e., 
follow transitions increase the 
potential function $\lambda(\cdot)$ given by Equation~\eqref{eq:potential}. 
This implies the following bound.
\begin{proposition}\label{prop:follow-terminates}
A deliberation that consists of single-agent transitions and follow transitions 
can have at most $n^2$ transitions.
\end{proposition}

However, in contrast to single-agent transitions, follow transitions are sufficient for successful
deliberation in any subset of the 1D Euclidean space.

\begin{theorem}\label{thm:follow-success-1D}
Consider a deliberation space $(X, V, r, \rho)$, where
$X, V\subseteq {\mathbb R}$ and $\rho(x, y)=|x-y|$.
Then every deliberation that consists of follow transitions only,
or of a combination of follow transitions and single-agent transitions,
is successful.
\end{theorem}
\begin{proof}
Assume for convenience that $r=0$. 
Consider a deliberation that consists of single-agent transitions and follow transitions.
By Proposition~\ref{prop:follow-terminates} we know that it is finite; let $\D$ be its terminal state. 

Suppose that $\D$ contains two deliberative coalitions $(C_1, p_1)$ and $(C_2, p_2)$ 
with $p_1, p_2\in {\mathbb R}^+$; assume without loss of generality 
that $p_1\le p_2$. Note that $C_1, C_2\subseteq {\mathbb R}^+$: all agents in ${\mathbb R}^-\cup\{0\}$
prefer $r$ to $p_1, p_2$. Furthermore, every agent in $C_2$
approves $p_1$: indeed, if $v\in C_2$ does not approve $p_1$ then 
$|v-r|\le |v-p_1|$, i.e.,  $v\le |v-p_1|$. Since $v, p_1>0$, this would imply $2v\le p_1\le p_2$, in which case
$v$ would not approve $p_2$ either.
Hence there is a follow transition in which
$C_2$ joins $C_1$, a contradiction with $\D$ being a terminal coalition structure.

Thus, $\D$ contains at most one coalition, say, $(C^+, q^+)$, that supports a proposal in ${\mathbb R}^+$;
by the same argument, it also contains at most one coalition, say, $(C^-, q^-)$, that 
supports a proposal in ${\mathbb R}^-$, and at most one coalition, say, $(C^0, r)$, that  
supports~$r$. 

Let $p$ be some proposal in $M^*$, and assume without loss of generality
that $p>0$.
Since agents in ${\mathbb R}^-$ prefer $r$ to $p$, we have $C^-\cap V^{p}=\emptyset$;
also, by definition, all agents in $C^0$ weakly prefer $r$ to $p$. Hence $V^p\subseteq C^+$
and therefore $|C^+|=m^*$.
\qed\end{proof}

If we modify the definition of follow transitions to require $|\d_2|\ge |\d_1|$
(i.e., that the joint proposal of the new coalition is the proposal originally 
supported by the larger of the two coalitions, a seemingly sensible condition), then the proof
of Theorem~\ref{thm:follow-success-1D} no longer goes through. In fact, Example~\ref{ex:single-bad}
illustrates that in this case the transition system may be unable to reach a successful state.\footnote{One might then also ask whether the constraint $|\d_2|\ge |\d_1|$ could be dropped for single-agent transitions (Definition~\ref{def:single}). However, allowing arbitrary moves by single agents (a)~does not capture the idea that agents prefer to be in larger coalitions, and (b)~may lead to cycles of transitions, and hence termination will no longer be guaranteed.}

Unfortunately, Theorem~\ref{thm:follow-success-1D} does not extend beyond one dimension.
The following examples show that in the Euclidean plane a deliberation that only uses
single-agent transitions and follow transitions is not necessarily successful. 

\begin{example}\label{ex:2d}
Consider a space of proposals $\{a, b, p, r\}$ embedded into ${\mathbb R}^2$, 
where~$r$ is located at $(0, 0)$, $p$ is located at $(0, 3)$, $a$ is located
at $(-3, 3)$, and $b$ is located at $(3, 3)$. There are four agents $v_1, v_2, v_3, v_4$
located at $(-3, 3)$, $(-3, 4)$, $(3, 3)$, and $(3, 4)$, respectively
(see Figure~\ref{figure:nofollow}).
Note that all agents prefer $p$ to $r$. 
Consider a deliberative coalition structure $\D = \{\d_1, \d_2\}$, 
where $\d_1=(\{v_1, v_2\}, a)$, $\d_2=(\{v_3, v_4\}, b)$.
Then no agent in $\d_1$ approves $b$, and no agent in $\d_2$
approves $a$, so there are no follow transitions and 
no single-agent transitions from $\D$.
\end{example}

\begin{figure}
	\centering
		\begin{tikzpicture}[scale=1]
	
	\node (r) at (0,0) [label=below:$r$] {};
	\fill (r) circle [radius=0.05];
	
	\node (p) at (0,1) [label=right:$\textcolor{blue}{p}$] {};
	\fill[blue] (p) circle [radius=0.05];
	\node (a) at (-1,1) [label=left:$\textcolor{blue}{a}$] {};
	\fill[blue] (a) circle [radius=0.05];
	\node (b) at (1,1) [label=right:$\textcolor{blue}{b}$] {};
	\fill[blue] (b) circle [radius=0.05];
	
	\node (v1) at (-1,1) [label=above:$\textcolor{red}{v_1}$] {};
	\fill[red] (v1) circle [radius=0.05];
	\node (v2) at (-1,1.7) [label=above:$\textcolor{red}{v_2}$] {};
	\fill[red] (v2) circle [radius=0.05];
	\node (v3) at (1,1) [label=above:$\textcolor{red}{v_3}$] {};
	\fill[red] (v3) circle [radius=0.05];
	\node (v4) at (1,1.7) [label=above:$\textcolor{red}{v_4}$] {};
	\fill[red] (v4) circle [radius=0.05];
	
	\node (W) at (-4,0) {};
	\node (E) at (4,0) {};
	\node (N) at (0,4) {};
	\node (S) at (0,-0.1) {};
	
	\draw[->] (W) -- (E);
	\draw[->] (S) -- (N);
    
    \draw[thick, dotted, color=blue] (1,1.7) ellipse (0.5cm and 1.3cm);
    \draw[thick, dotted, color=blue] (-1,1.7) ellipse (0.5cm and 1.3cm);
	
	\end{tikzpicture}
	\caption{The metric space in Example~\ref{ex:2d}.}
	\label{figure:nofollow}
\end{figure}

\begin{example}\label{ex:circle}
Consider the example in Figure~\ref{fig:circle}, where all proposals other than $r$
as well as all agents are located on a circle with center $r$; the metric is the usual
Euclidean distance in ${\mathbb R}^2$. 
The alternatives $a$, $b$, and $c$ are equidistant from each other, 
and $p$ is located halfway between $a$ and~$b$ (so, in particular, $d(a, p) = d(p, b) = d(p, r)$).
The deliberative coalition structure is $\D = (\d_1, \d_2, \d_3)$, 
where $\d_1=(\{v_1, v_2, v_3\}, a)$, $\d_2=(\{v_4, v_5, v_6\}, b)$, and $\d_3=(\{v_7\}, c)$;
one can verify that each agent approves the proposal supported by her deliberative coalition. 
Observe that more than half of the agents (namely, $v_2, v_3, v_4, v_5$) prefer $p$ to $r$. 
However, there are no follow transitions or single-agent transitions from $\D$:
for every pair of distinct coalitions $\d_i, \d_j$ the agents in $\d_i$ do not approve
the proposal supported by $\d_j$.
\end{example}

\begin{figure}
	\centering
	\begin{tikzpicture}
	
	\node (r) at (0,0) [label=below:$r$] {};
	\fill (r) circle [radius=0.05];
	\draw (0,0) circle (2cm);
	
	\node (c) at (-2,0) [label=left:$\textcolor{blue}{c}$] {};
	\fill[blue] (c) circle [radius=0.05];
	
	\node (p) at (2,0) [label=right:$\textcolor{blue}{p}$] {};
	\fill[blue] (p) circle [radius=0.05];
	
	\node (a) at (1,1.73) [label=right:$\textcolor{blue}{a}$] {};
	\fill[blue] (a) circle [radius=0.05];
	
	\node (v1) at (0,2) [label=below:$\textcolor{red}{v_1}$] {};
	\fill[red] (v1) circle [radius=0.05];
	
	\node (v2) at (1.73,1) [label=right:$\textcolor{red}{v_2}$] {};
	\fill[red] (v2) circle [radius=0.05];
	
	\node (v3) at (1.92,0.5) [label=left:$\textcolor{red}{v_3}$] {};
	\fill[red] (v3) circle [radius=0.05];
	
	\draw[thick, dotted, color=blue, rotate=-30] (0.3,1.7) ellipse (1.5cm and 1cm);
	
	\node (v7) at (-1.92,-0.5) [label=left:$\textcolor{red}{v_7}$] {};
	\fill[red] (v7) circle [radius=0.05];
	
	\draw[thick, dotted, color=blue] (-2.1,-0.3) ellipse (0.5cm and 0.75cm);
	
	\node (v5) at (1.73,-1) [label=right:$\textcolor{red}{v_5}$] {};
	\fill[red] (v5) circle [radius=0.05];
	
	\node (v4) at (1.92,-0.5) [label=left:$\textcolor{red}{v_4}$] {};
	\fill[red] (v4) circle [radius=0.05];
	
	\node (b) at (1,-1.73) [label=above:$\textcolor{blue}{b}$] {};
	\fill[blue] (b) circle [radius=0.05];
	
	\node (v6) at (0,-2) [label=above:$\textcolor{red}{v_6}$] {};
	\fill[red] (v6) circle [radius=0.05];
	
	\draw[thick, dotted, color=blue, rotate=30] (0.3,-1.7) ellipse (1.5cm and 1cm);
	
	\draw[dotted] (a)--(b);
	\draw[dotted] (b)--(c);
	\draw[dotted] (a)--(c);
	
\end{tikzpicture}
	\caption{The metric space in Example \ref{ex:circle}.}
	\label{fig:circle}
\end{figure}


\section{Merge Transitions}\label{section:merge}

So far, we have focused on transitions that did not introduce new proposals.
Examples~\ref{ex:2d} and~\ref{ex:circle}, however, show that new proposals may be necessary 
to reach success: indeed, none of the proposals supported by existing coalitions 
in these examples was approved by a majority of agents. Thus, next we explore transitions
that identify new proposals. 

As a first step, it is natural to relax the constraint in the definition
of the follow transitions that requires the new coalition to 
adopt the proposal of one of the two component coalitions, 
and, instead, allow the agents to identify a new proposal
that is universally acceptable.  

We do not specify how the compromise proposal $p$ is identified. 
We imagine that a new proposal could be
put forward by one
of the agents in $\d_1, \d_2$ or by an external
mediator whose goal is 
to help the agents reach a consensus.
%

\begin{definition}[Merge Transition]
\label{def:merge}
A pair of coalition structures $(\D, \D')$ forms a {\em merge transition}
if there exist non-empty coalitions $\d_1, \d_2\in \D$ and $\d'_2\in \D'$
such that $\d_1=(C_1, p_1)$, $\d_2=(C_2, p_2)$,
$\D\setminus\{\d_1, \d_2\}=\D'\setminus\{\d'_2\}$, 
and $\d'_2 = (C_1\cup C_2, p)$ for some proposal $p$. 
\end{definition}
One can verify that in Example~\ref{ex:2d} the agents have a merge transition 
to the deliberative coalition $(\{v_1, v_2, v_3, v_4\}, p)$. Thus, merge transitions 
are strictly more powerful than follow transitions. 
Moreover, by considering the potential function $\lambda(\cdot)$
defined in Equation \eqref{eq:potential}, one can check that 
a deliberation that consists of single-agent transitions and merge transitions 
can have at most $n^2$ steps. This implies the following bound.
\begin{proposition}\label{prop:merge-terminates}
A deliberation that consists of single-agent transitions and merge transitions 
can have at most $n^2$ transitions.
\end{proposition}

We will now describe a class of metric spaces where merge transitions guarantee convergence, but follow transitions do not.

\paragraph{Trees}
Consider a weighted tree $T$ with a set of vertices $U$, a set of edges $E$, and a weight function $\omega: E\to {\mathbb R}^+$. This tree defines a metric space $(M_T, \rho_T)$, where $M_T=U$, and for a pair of vertices $x, y\in U$ the distance $\rho_T(x, y)$ is the length of the path between $x$ and $y$ in $T$, where the length of an edge $e$ is given by $\omega(e)$.
Consequently, a tree $T$ defines a family of deliberation spaces $(X, V, r, \rho)$, where $X\subseteq M_T$, $V\subseteq M_T$, $r\in X$ and $\rho=\rho_T$. It turns out that for any such deliberation space merge transitions guarantee successful termination. 

\begin{theorem}\label{thm:merge-success-tree}
Consider a deliberation space $(X, V, r, \rho)$ that corresponds
to a tree $T=(U, E, \omega)$ with $r\in U$.
Then every deliberation that consists of merge transitions only,
or of a combination of merge transitions and single-agent transitions, is successful.
\end{theorem}
\begin{proof}
It will be convenient to think of $T$ as a rooted tree with root $r$ and, for each $x\in U$, denote by $T_x$ the subtree of $T$ that is rooted in $x$.

Assume for convenience that $X=M_T$, i.e., every vertex of $T$
is a feasible proposal (later we will show that this assumption is not necessary). Let $p$ be some proposal in $M^*$, and let $q$ be the child of $r$ such that the path from $p$ to $r$ passes through $q$ (it is possible that $p=q$). We claim that $q\in M^*$. Indeed, an agent approves $q$ if and only if her ideal point is located in~$T_q$. On the other hand, if an agent's ideal point is not located in $T_q$, then the path from her ideal point to $p$ passes through the root $r$ and hence she does not approve $p$. Hence, the set of agents who approve $p$ forms a subset of the set of agents who approve $q$, establishing our claim.

This argument also shows that there are exactly $m^*$ agents whose
ideal point is located in $T_q$. Each such agent can only be in a deliberative coalition that supports a proposal in $T_q$: indeed, if an agent's ideal point is in $T_q$, then she does not approve proposals outside of $T_q$.

Now, consider a terminal coalition structure $\D$ and suppose that there does not exist a deliberative coalition of size $m^*$ supporting a proposal in $T_q$. Then, there are multiple deliberative coalitions that support proposals in $T_q$ and consist of agents whose ideal points are in $T_q$. Let $(C, p')$ and $(C, p'')$ be two such coalitions. But then there is a merge transition in which coalition $(C'\cup C'', q)$ forms, a contradiction with $\D$ being a terminal coalition structure.

If $X\neq M_T$, then we can use a similar argument; the only modification is that we have to define $q$ by considering the ancestors of $p$ (other than $r$) that belong to $X$, and, among these, pick one that is closest to $r$. \qed
\end{proof}

Note that, for trees, we do need the full power of merge transitions: 
follow or single-agent transitions are insufficient for successful deliberation. 
This holds even if all edges have the same weight, as demonstrated by the following example.

\begin{example}\label{ex:no-follow}
Consider a tree $T$ with vertex set $U = \{r, u_1, u_2, u_3\}$
and edge set $E = \{\{r, u_1\}, \{u_1, u_2\}, \{u_1, u_3\}\}$;
the weight of each edge is $1$. Let $X=U$, and suppose that there are two agents 
with ideal points at $u_2$ and $u_3$, respectively. 
Consider a coalition structure where each of these agents forms 
a singleton coalition around her ideal point. Note that 
$\rho(u_2, r)=\rho(u_3, r)=\rho(u_2, u_3)$ and hence the agents 
do not approve each other's positions. 
Therefore, there are no follow transitions or single-agent transitions 
from this coalition structure, even though $u_1$ is approved by both agents 
and therefore there is a merge transition where both agents form 
a coalition around $u_1$. 
\end{example}

In contrast, the next example shows that, for Euclidean spaces, merge transitions are insufficient for successful deliberation, even when combined with single-agent transitions.

\begin{example}\label{ex:needcomp}
Consider a modification of Example~\ref{ex:2d}, where we add an
agent $v_5$ at $(-4, 0)$ to $\d_1$ and an agent $v_6$ at $(4, 0)$ to $\d_2$;
let $\D'$ be the resulting coalition structure.
Note that $v_5$ approves $a$, but does not approve $p$, thereby 
preventing a merge transition where agents in $\d_1$ and $\d_2$ 
form a coalition around $p$. 
\end{example}
The same conclusion can be drawn from Example~\ref{ex:circle}. Indeed, in that example each coalition contains an agent who does not approve of any proposals except for the one supported by her coalition ($v_1$ does not approve $p$, $b$ and $c$, $v_6$ does not approve $p$, $a$ and $c$, and $v_7$ does not approve $a$, $b$ and $p$). Hence, no merge transition is possible.


\section{Compromise Transitions}\label{section:compromise}

Examples~\ref{ex:circle} and~\ref{ex:needcomp} illustrate that, to reach a successful outcome, 
coalitions may need to leave some of their members behind when joining forces.
The following definition formalizes this idea. 
%

\begin{definition}[Compromise Transition]
\label{def:compromise}
A pair of coalition structures $(\D, \D')$ forms a {\em compromise transition}
if there exist coalitions $\d_1, \d_2\in \D$, $\d'_1, \d'_2, \d'\in \D'$
and a proposal $p$ 
such that 
$\D\setminus\{\d_1, \d_2\}=\D'\setminus\{\d'_1, \d'_2, \d'\}$, 
$\d_1=(C_1, p_1)$, $\d_2=(C_2, p_2)$,
$\d'=(C_1^p\cup C_2^p, p)$, 
$\d'_1=(C_1\setminus C_1^p, p_1)$, $\d_2=(C_2\setminus C_2^p, p_2)$,
and $|C_1^p\cup C_2^p|>|C_1|, |C_2|$.
%
\end{definition}

Intuitively, under a compromise transition some of the agents in $\d_1$ and $\d_2$
identify a suitable proposal $p$, 
and then those of them who approve $p$
move to form a coalition that supports $p$, leaving the rest of their old friends behind;
a necessary condition for the transition is that the new coalition
should be larger than both $\d_1$ and $\d_2$.

\begin{example}\label{ex:2d-compromise}
In Example~\ref{ex:circle}, there is a compromise transition from $\D$:
agents $v_2$, $v_3$, $v_4$ and $v_5$ form
a coalition around $p$, with $v_1$
and $v_6$ (as well as $v_7$) remaining in singleton coalitions.
Similarly, in Example~\ref{ex:needcomp}
there is a compromise transition from $\D'$
to the coalition structure $\{\d'_1, \d'_2, \d'\}$, 
where $\d'_1 = (\{v_5\}, a)$, $\d'_2 = (\{v_6\}, b)$, and 
$\d'=(\{v_1, v_2, v_3, v_4\}, p)$.
In either case the new coalition has size $4$, so the resulting coalition structure is successful.
\end{example}

Importantly, we assume that all agents in $\d_1$ and $\d_2$ who approve $p$ join the compromise
coalition; indeed, this is what we expect to happen if the agents
myopically optimize the size of their coalition.


An important feature of compromise transitions 
is that they ensure termination. 

\begin{proposition}\label{prop:lex}
A deliberation that consists of compromise transitions 
can have at most $n^n$ transitions.
\end{proposition}
\begin{proof}
Consider a coalition structure $\D=\{\d_1, \dots, \d_m\}$ such that $\d_i=(C_i, p_i)$
for each $i\in [m]$. Assume without loss of generality that $|C_1|\ge\dots\ge|C_m|$ and there exists an $\ell\in [m]$
such that $|C_i|>0$ for $i\in [\ell]$, $|C_i|=0$ for $i=\ell+1, \dots, m$. 
Let $\gamma(\D)=(|C_1|, \dots, |C_\ell|)$.
Note that $\gamma(\D)$ is a non-increasing sequence of positive integers. Given two non-increasing sequences $(a_1, \dots, a_s)$, $(b_1, \dots, b_t)$
of positive integers we write
\begin{align} \label{eq:lexdef}
(a_1, \dots, a_s) <_{\mathrm{lex}} (b_1, \dots, b_t)
\end{align}
if either 
(a) there exists a $j\le \min\{s, t\}$ such that $a_i=b_i$ for all $i<j$ and 
$a_j<b_j$, or 
(b) $s<t$, and $a_i=b_i$ for all $i\in [s]$.
Note that $<_{\mathrm{lex}}$ is a total order on the space of non-increasing
sequences of positive integers.
Now, observe that if $(\D, \D')$ is a compromise transition, 
then for the respective coalitions $\d_1, \d_2, \d'$ we have 
$|\d'|>\max\{|\d_1|, |\d_2|\}$, and hence
\begin{equation}\label{eq:lex}
\gamma(\D)<_{\mathrm{lex}}\gamma(\D').
\end{equation} 
Now, consider a sequence of deliberative coalition structures $\D_1, \D_2, \dots$ 
such that for each $i\ge 1$ the pair $(\D_i, \D_{i+1})$ is a compromise transition. 
Since $<_{\mathrm{lex}}$ is transitive, 
it follows from Equation~\ref{eq:lex} that $\gamma(\D_i)\neq \gamma(\D_j)$ for each $i<j$.
As each $\gamma(\D_i)$ is a sequence of at most $n$ numbers, with each number taking values 
between $1$ and $n$, our claim follows.
\qed\end{proof}

In contrast to the case of
single-agent transitions and follow/merge transitions, we are unable to show
that a deliberation consisting of compromise transitions 
always terminates after polynomially many steps\footnote{
 In a follow-up paper, \citet{followup} have shown that a deliberation consisting
 of compromise transitions may be of length $\Omega(2^{\sqrt{n}/2})$. They have also
 improved the upper bound in Proposition~\ref{prop:lex} to $2^n$ by designing a suitable potential function.}.
We note that a compromise
transition does not necessarily increase the 
potential function $\lambda(\cdot)$ given by Equation~\eqref{eq:potential}: 
e.g., in Example~\ref{ex:needcomp} the transition where $v_1, v_2, v_3, v_4$ 
form a coalition around $p$ does not change the value of $\lambda$

We say that a deliberation space $(X, V, r, \rho)$ 
is a {\em Euclidean deliberation space} if
$X={\mathbb R}^d$, $V\subseteq {\mathbb R}^d$ for some $d\ge 1$,
and $\rho$ is the standard Euclidean metric on ${\mathbb R}^d$.
The main result of this section is that, in every Euclidean
deliberation space, every maximal run of compromise transitions is successful.
To prove it, we need two auxiliary lemmas. 
In what follows, for a coalition structure $\D$,
$|\D|$ denotes the number of non-empty coalitions in $\D$ that do not support $r$.

\begin{lemma}\label{lem:2}
In every deliberation space, a deliberation that consists
of compromise transitions and has a coalition structure $\D$
with $|\D|=2$ as its terminal state is successful. 
\end{lemma}
\begin{proof}
 Consider a coalition structure $\D$ with $|\D|=2$, and suppose that $\D$ is not successful.
 Let $\d_1$, $\d_2$ be the two non-empty coalitions in $\D$ that do not support $r$;
 we have $|\d_1|<m^*$, $|\d_2|<m^*$. 
 For each $p\in M^*$ we have $V^p=\d_1^p\cup \d_2^p$ and hence
 $|\d_1^p\cup \d_2^p|=m^* >|\d_1|, |\d_2|$.
 Thus, there exists a compromise transition from $\D$ in which
 agents in $\d_1^p\cup \d_2^p$ form a coalition around $p$. 
\qed
\end{proof}

\begin{figure}
    \centering

\begin{tikzpicture}[scale=1.3]

\node (a) at (2,1) [label=above:$\textcolor{red}{a}$] {};
\node (r) at (2,0) [label=above right:$r$] {};
\node (l) at (4,0.5) [label=right: $\ell$] {};
\node (l') at (4,0) [label=right: $\ell'$] {};
\node (l1) at (4,0.5)  {};
\node (-l1) at (0,-0.5) [label=left: $\ell_1'$] {};
\node (l2) at (4,-0.3)  {};
\node (-l2) at (0,0.3) [label=left: $\ell_2'$] {};

\node at (3,1) {$\textcolor{red}{C_a}$};
\node at (1,-1.5) {$\textcolor{green}{C_b}$};
\node at (3,-1.5) {$\textcolor{blue}{C_c}$};

\node (p1) at (0,-1) {};
\node (p2) at (4,-0.55) {};

\node (p) at (2.1,-0.77) [label= below: $r''$] {};

\draw[dotted, thick] (l1) -- (-l1);
\draw[dotted, thick] (l2) -- (-l2);
\draw[thick] (p1) -- (p2);

\draw (0,0) -- (4,0);
\draw (0,0.5) -- (4,0.5);
\draw[dotted] (a) -- (r);
\draw[dotted] (r) -- (p);

\draw[color=red] (3,1) ellipse (0.75cm and 0.4cm);
\draw[color=green] (1,-1.5) ellipse (0.75cm and 0.4cm);
\draw[color=blue] (3,-1.5) ellipse (0.75cm and 0.4cm);

\fill (r) circle [radius=0.05];
\fill (a) circle [radius=0.05];
\fill (p) circle [radius=0.05];

\draw[gray!50,ultra thick] (1.5,-0.1) arc (211:337:5mm) ;
 
\end{tikzpicture}

    \caption{Proof of Lemma~\ref{lem:rd-amicable}.}
    \label{fig:euclidean}
\end{figure}

\begin{lemma}\label{lem:rd-amicable}
In every Euclidean deliberation space, a deliberation that consists
of compromise transitions and has a coalition structure $\D$
with $|\D|\ge 3$ as its terminal state is successful. 
\end{lemma}

\begin{proof}
For the case $d=1$ our claim follows from the proof of Theorem~\ref{thm:follow-success-1D}.
We will now provide a proof for $d=2$; it generalizes straightforwardly to $d>2$.

Fix a coalition structure $\D$ that contains at least three non-successful coalitions none of which support $r$; we will show that $\D$ is not terminal.
Let $\d_a=(C_a, a)$ be a maximum-size coalition in $\D$ that does not support $r$. 
Let $\ell$ be the perpendicular bisector to the $a$--$r$ segment, i.e., 
$\ell$ passes through the middle of the $a$--$r$ segment and is orthogonal to it. 
Note that $\ell$ separates ${\mathbb R}^2$ into two open half-planes,  
so that $r$ lies in one of these 
half-planes, while all points in $C_a$ lie in the other half-plane (see Figure~\ref{fig:euclidean}). 
Let $\ell'$ be the line that passes through~$r$ and is parallel to $\ell$. 
For a positive $\alpha$, let $\ell'_1$ be the line obtained by rotating $\ell'$ about $r$ 
clockwise by $\alpha$, and let $\ell'_2$ be the line obtained by rotating $\ell'$ counterclockwise 
by $\alpha$. The line $\ell'_1$ (respectively, $\ell'_2$) partitions ${\mathbb R}^2$ into open half-planes 
$H_1$ and $H'_1$ (respectively, $H_2$ and $H'_2$). We can choose $\alpha$ to be small enough so that 
$C_a\subset H_1, C_a\subset H_2$ and so that no agent lies on $\ell'_1$ or on $\ell'_2$.

Now, if there exists a coalition $\d_b=(C_b, b)\in \D$, $\d_b\neq \d_a$, $b\neq r$, 
such that $v\in H_1$ or $v\in H_2$ 
for some $v\in C_b$, then $r$ is not in the convex hull of $C_a$ and $v$, and hence 
there is a line that separates $C_a\cup\{v\}$ from $r$; by projecting $r$ onto this line, 
we obtain a proposal $r'$ that is approved by $v$ and all agents in $C_a$.
Thus, there is a compromise transition in which a non-empty subset of agents in $C^{r'}_b$
joins $C_a$ to form a deliberative coalition around $r'$.

Otherwise, $\d_x\subseteq H'_1\cap H'_2$ for all $\d_x\in \D\setminus\{\d_a\}$.
Consider two distinct coalitions $\d_b, \d_c\in \D$ with 
$\d_b=(C_b, b)$, $\d_c=(C_c, c)$ and $b, c\neq r$. 
As $H'_1\cap H'_2$ is bounded by two rays 
that start from $r$, and the angle between these rays is $2\pi-2\alpha<2\pi$, 
there is a line $\ell^*$ that divides ${\mathbb R}^2$ into two open half-planes 
so that $r$ is in one half-plane and the set $C_b\cup \C_c$ is in the other half-plane; thus, all agents in $C_b\cup C_c$ approve the proposal $p$ 
obtained by projecting $r''$ onto~$\ell^*$, and hence there is a merge transition.
\qed\end{proof}

Combining Lemmas~\ref{lem:2} and~\ref{lem:rd-amicable}, we obtain the main result
of this section.

\begin{theorem}\label{thm:success}
  In every Euclidean deliberation space, every maximal run of compromise transitions is a successful deliberation.
\end{theorem}
\begin{proof}
 Consider a maximal run of compromise transitions, and let $\D$ be its terminal state. Suppose for the sake of contradiction that $\D$ is not successful.
 Note that this implies that $|\D|>1$. If $|\D|=2$ there is a transition from $\D$
 by Lemma~\ref{lem:2} and if $|\D|\ge 3$, then there is a transition from $\D$
 by Lemma~\ref{lem:rd-amicable}. 
\qed\end{proof}

For the proof of Lemma~\ref{lem:rd-amicable} to go through, 
the underlying metric space should be sufficiently rich: to obtain the proposal 
approved by the new coalition, we project the status quo $r$ on a certain line.
The argument goes through if we replace ${\mathbb R}^d$
with ${\mathbb Q}^d$; however, it does not extend to the case where $X$ is an arbitrary
finite subset of ${\mathbb R}^d$. 
Intuitively, for deliberation to converge, at least some agents
should be able to spell out nuanced compromise proposals. 

Observe that the compromise transitions in Lemma~\ref{lem:rd-amicable} have a special form:
when two coalitions join forces, at least one of them is fully behind the new proposal.
This motivates the following definition.

\begin{definition}[Subsume Transition]
\label{def:subsume}
A pair of coalition structures $(\D, \D')$ forms a {\em subsume transition}
if there exist coalitions $\d_1, \d_2\in \D$, $\d'_1, \d'\in \D'$
and a proposal $p$ 
such that 
$\D\setminus\{\d_1, \d_2\}=\D'\setminus\{\d'_1, \d'\}$, 
$\d_1=(C_1, p_1)$, $\d_2=(C_2, p_2)$, $C_2^p=C_2$, $C_1^p\neq\emptyset$,
$\d'=(C_1^p\cup C_2, p)$, 
$\d'_1=(C_1\setminus C_1^p, p_1)$, and $|C_1^p\cup C_2|>|C_1|$.
\end{definition}

By construction, every subsume transition is a compromise transition.
Since every merge transition is also a subsume transition (with $\d'_1$ being empty), it follows
that Lemma~\ref{lem:rd-amicable} holds for deliberations 
that consist of subsume transitions.


While subsume transitions by themselves are not sufficient
for successful deliberation in ${\mathbb R}^d$ (as the transition in 
Lemma~\ref{lem:2} is not necessarily a subsume transition), they are nearly sufficient: the proof of Theorem~\ref{thm:success}
shows that we need at most one general compromise transition.
Since every subsume transition 
increases the potential defined by Equation~\eqref{eq:potential}
by 
$$
(|C_2|+|C_1^p|)^2+(|C_1|-|C_1^p|)^2-(|C_1|^2+|C_2|^2)=
2|C_1^p|\cdot(|C_2|-|C_1|+|C_1^p|)\ge 2,
$$
we obtain the following corollary.

\begin{corollary}
For every Euclidean deliberation space 
there exists a successful deliberation that consists of compromise transitions
and has length at most $n^2+1$.
\end{corollary}
\begin{proof}
Suppose that agents perform subsume transitions until no such transitions are available; 
as every subsume transition increases the potential $\lambda(\cdot)$ (see Equation~\ref{eq:potential}), 
this process
ends after at most $n^2$ steps. By Lemma~\ref{lem:rd-amicable}, if the resulting
coalition structure $\D$ is not successful, then $|\D|=2$, in which case,
by Lemma~\ref{lem:2}, there is a compromise transition from $\D$ to a successful
coalition structure. 
\qed
\end{proof}


Given our positive results for Euclidean deliberation spaces, 
it is natural to ask whether compromise 
transitions are sufficient for convergence in other metric spaces.
The following example, however, illustrates that this is not the case
even if $X$ is a finite subset of ${\mathbb R}^d$.

\begin{figure}
	\centering
		\begin{tikzpicture}[scale=1.3]
	
	\node (r) at (0,0) [label=below:$r$] {};
	\fill (r) circle [radius=0.05];
	
	\node (p) at (0,2) [label=right:$\textcolor{blue}{p}$] {};
	\fill[blue] (p) circle [radius=0.05];
	
	\node (e) at (0,3.5) [label=right:$\textcolor{blue}{e}$] {};
	\fill[blue] (e) circle [radius=0.05];
	
	\node (a) at (2,0) [label=below:$\textcolor{blue}{a}$] {};
	\fill[blue] (a) circle [radius=0.05];
	\node (b) at (1.5,0.6) [label=below:$\textcolor{blue}{b}$] {};
	\fill[blue] (b) circle [radius=0.05];
	\node (c) at (-2,0) [label=below:$\textcolor{blue}{c}$] {};
	\fill[blue] (c) circle [radius=0.05];
	\node (d) at (-1.5,-0.6) [label=below:$\textcolor{blue}{d}$] {};
	\fill[blue] (d) circle [radius=0.05];
	
	\node (v1) at (3,0) [label=above:$\textcolor{red}{v_1}$] {};
	\fill[red] (v1) circle [radius=0.05];
	\node (v2) at (2.5,1) [label=above:$\textcolor{red}{v_2}$] {};
	\fill[red] (v2) circle [radius=0.05];
	\node (v3) at (-3,0) [label=above:$\textcolor{red}{v_3}$] {};
	\fill[red] (v3) circle [radius=0.05];
	\node (v4) at (-2.5,-1) [label=above:$\textcolor{red}{v_4}$] {};
	\fill[red] (v4) circle [radius=0.05];
	
	\node (v5) at (2,2) [label=above:$\textcolor{red}{v_5}$] {};
	\fill[red] (v5) circle [radius=0.05];
	\node (v6) at (1.5,2.6) [label=above:$\textcolor{red}{v_6}$] {};
	\fill[red] (v6) circle [radius=0.05];
	\node (v7) at (-2,2) [label=above:$\textcolor{red}{v_7}$] {};
	\fill[red] (v7) circle [radius=0.05];
	\node (v8) at (-1.5,1.4) [label=above:$\textcolor{red}{v_8}$] {};
	\fill[red] (v8) circle [radius=0.05];

	\node (v9) at (0,3) [label=right:$\textcolor{red}{v_9}$] {};
	\fill[red] (v9) circle [radius=0.05];
	
	\node (W) at (-4,0) {};
	\node (E) at (4,0) {};
	\node (N) at (0,4) {};
	\node (S) at (0,-0.1) {};
	\node (SW) at (-3.5,-1.4) {};
	\node (NE) at (3.5,1.4) {};
	
	\node (W1) at (-3,2) {};
	\node (E1) at (3,2) {};
	\node (SW1) at (-2.5,1) {};
	\node (NE1) at (2.5,3) {};

	\draw[->] (W) -- (E);
	\draw[->] (S) -- (N);
	\draw[->] (SW) -- (NE);
	
	\draw[dotted] (W1) -- (E1);
	\draw[dotted] (SW1) -- (NE1);
	\draw[dotted] (a) -- (v5);
	\draw[dotted] (b) -- (v6);
	\draw[dotted] (c) -- (v7);
	\draw[dotted] (d) -- (v8);
	
	\end{tikzpicture}
	\caption{The deliberation space in Example~\ref{ex:nocomp}.}
	\label{fig:nocomp}
\end{figure}

\begin{example}\label{ex:nocomp}
Figure~\ref{fig:nocomp} depicts a deliberation space that is embedded in ${\mathbb R}^3$,
so that $\rho$ is the usual Euclidean metric in ${\mathbb R}^3$, 
and $X=\{r, a, b, c, d, p, e\}$, $V=\{v_1, \dots, v_9\}$ with
\begin{align*}
&r=(0, 0, 0), \\
&a=(2, 0, 0), b=(0, 2, 0), c=(-2, 0, 0), d=(0, -2, 0), \\
&p = (0, 0, 2), e=(0, 0, 3.5),\\
&v_1=(3, 0, 0), v_2=(0, 3, 0), v_3=(-3, 0, 0), v_4=(0, -3, 0), \\
&v_5=(2, 0, 2), v_6=(0, 2, 2), v_7=(-2, 0, 2), v_8=(0, -2, 2), \\
&v_9 = (0, 0, 3).
\end{align*}
Then $\D=\{\d_1, \d_2, \d_3, \d_4, \d_5\}$, where
$\d_1=(\{v_1, v_5\}, a)$, 
$\d_2=(\{v_2, v_6\}, b)$, 
$\d_3=(\{v_3, v_7\}, c)$, 
$\d_4=(\{v_4, v_8\}, d)$, 
$\d_5=(\{v_9\}, e)$ is a deliberative coalition structure, 
and there are no compromise transitions from $\D$.
However, $\D$ is not successful, as agents $v_5, \dots, v_9$ all approve $p$.
\end{example}  

\noindent Thus, for general metric spaces one needs an even richer class of transitions
to identify credible alternatives to the status quo.


\section{Beyond Two-Way Compromises?}\label{section:hypercube}

We closed the previous section by showing that `sparse' subspaces of Euclidean spaces may be challenging even for compromise transitions (Example \ref{ex:nocomp}). In this section we explore other `sparse' spaces, but move away altogether from the Euclidean framework. Specifically, we focus on spaces that arise naturally in decision-making on combinatorial domains \citep{ma-comsoc}: $d$-dimensional hypercubes. These consist of the set $\{0, 1\}^d$ of binary vectors of length $d$ ($d$ being a positive integer), endowed with the discrete metric 
$h$, defined as 
$$
h((x_1, \dots, x_d), (y_1, \dots, y_d))=\sum_{i=1}^d|x_i-y_i|.
$$
This metric is known as the Hamming, or Manhattan, distance. Intuitively, each element of $\{0, 1\}^d$ denotes one possible position on a set of $d$ binary issues. For any positive integer $d$, we refer to the space $(\{0, 1\}^d, h)$ as the \emph{$d$-dimensional hypercube} or simply \emph{$d$-hypercube}.
In what follows, for readability, we will often write vectors in $\{0, 1\}^d$
as strings of 1s and 0s of length $d$: e.g., we will write $0110$
instead of $(0, 1, 1, 0)$. 
Also, without loss of generality we will assume throughout that 
the status quo $r$ is the all-0 vector. 

\subsection{Deliberation on $d$-Hypercubes}

We start with a simple observation:
\begin{proposition}\label{prop:hdone}
In $d$-hypercubes with $d \leq 2$, every deliberation consisting of compromise transitions is successful. 
\end{proposition}
\begin{proof}
For $d = 1$ the claim holds trivially. Indeed, a terminal coalition structure cannot contain
two non-empty coalitions that both support the same proposal $a\neq r$ (as they could then merge
and form a bigger coalition around $a$). Hence, since
the deliberation space only contains one point other than $r=0$, all agents that 
support $1$ must be in the same coalition.

For $d = 2$, assume towards a contradiction that there exists a terminal deliberative 
coalition structure $\D$ that is not successful. Let $p$ be a most approved alternative, 
$p\neq 00$, and suppose that $p$ is approved by $m^*$ agents. 
Note that we can choose $p\neq 11$, as all agents who approve
$11$ also approve both $10$ and $01$; assume without loss of generality that $p=10$.

Since $\D$ cannot contain two non-empty coalitions that support the same proposal $a\neq r$, 
it contains at most three non-empty coalitions that do not support $00$. 
Suppose first that $\D$ contains a deliberative coalition $(C, 11)$ with $C\neq\emptyset$. 
Note that all agents in $C$ have $11$ as their ideal point, as agents whose ideal point 
is $01$ or $10$ do not approve $11$. Since $|C|<m^*$, there 
must exist a coalition $(C', q')$ with $C'\neq\emptyset$ and $q'=01$ or $q'=10$. 
But then there is a follow transition in which $C$ joins $C'$, so $\D$ is be terminal. 

Hence, $\D$ contains two non-empty deliberative coalitions: 
$(C_{10}, 10)$ and $(C_{01}, 01)$. Moreover, $C_{01}$ contains some agents 
whose ideal point is $11$ (as otherwise $C_{10}$ would contain all agents who approve $p=10$, 
a contradiction with $\D$ not being successful). But then there is a compromise transition 
where all agents in $C_{01}$ whose ideal point is $11$ move to $(C_{10}, 10)$, to form a coalition
of size $m^*>|C_{01}|$, a contradiction.
\qed
\end{proof}

Observe that the proof of Proposition~\ref{prop:hdone} holds for subsume
transitions: indeed, the proof shows that if a coalition structure $\D$
is not successful, then there is a subsume transition from it.
However, we will now show that other types of transitions considered in our work do not necessarily result in successful deliberation on $2$-hypercubes.
\begin{remark}
Deliberations that consist of single-agent, follow, or merge transitions may fail to be successful on $2$-dimensional hypercubes. Indeed, suppose that $r=00$, 
$v_1=v_2=10$, $v_3=01$, and $v_4=v_5=11$.
Note that $10$ is approved by agents $v_1$, $v_2$, $v_4$, and $v_5$, 
$01$ is approved by agents $v_3$, $v_4$, and $v_5$, and $11$
is approved by agents $v_4$ and $v_5$, so the most approved
proposal is $10$.
Consider the deliberative coalition structure $(\d_1, \d_2)$, 
where $\d_1$ consists of $v_1$ and $v_2$, and supports $10$, whereas $\d_2$ consists of $v_3$, $v_4$, and $v_5$, and supports $01$. There are no single-agent, follow or merge transitions from this coalition structure, even though $01$ 
is not the most approved project.
\end{remark}


In higher dimensions, however, compromise transitions no longer guarantee success.
\begin{example}[Compromise failure on $3$-hypercube] \label{ex:3hyper}
Consider the $3$-hypercube $\{ 0,1 \}^3$ and let $N$ consist of the following agents: 
$v_1 = 100$, $v_2 = 010$, $v_3 = 110$, $v_4 = 001$, $v_5 = 101$ (see Figure~\ref{fig:3hyper}). 
Let $\D$ be the deliberative coalition structure that contains the following three coalitions: $(\{ v_1 \}, 100)$, $(\{ v_2, v_3 \}, 010)$, and $(\{ v_4, v_5 \}, 001)$. This coalition structure is terminal. It is not successful, however, as $v_1$, $v_3$ and $v_5$ approve $100$. 
\end{example}
Consequently, single-agent, follow, merge and subsume transitions also fail to guarantee success in the above setting. However, a careful inspection of Example~\ref{ex:3hyper} suggests that a successful outcome would be within reach if compromises that involve three coalitions were allowed. This motivates us to study a natural multi-party generalization of compromise transitions.  

\begin{figure}
	\centering
	\begin{tikzpicture}[scale=1.3]
	
	\node (000) at (0,0) [label=below:$000~ r$] {};
	\fill (000) circle [radius=0.05];
	
	\node (011) at (1.5,2.6) [label=above:$011$] {};
	
	\node (010) at (0,2) [label=above:$010~ \textcolor{red}{v_2}$] {};
	\fill[red] (010) circle [radius=0.05];
	
	\node (001) at (1.5,0.6) [label=left:$001~ \textcolor{red}{v_4}$] {};
	\fill[red] (001) circle [radius=0.05];
	
	\node (111) at (3.5, 2.6) [label=right:$111$] {};
	
	\node (110) at (2,2) [label=above:$110~ \textcolor{red}{v_3}$] {};
	\fill[red] (110) circle [radius=0.05];
	
	\node (101) at (3.5,0.6) [label=right:$101~ \textcolor{red}{v_5}$] {};
	\fill[red] (101) circle [radius=0.05];
	
	\node (100) at (2,0) [label=below:$100~ \textcolor{red}{v_1}$] {};
	\fill[red] (100) circle [radius=0.05];
	
	\draw[thick, dotted, color=blue] (2.5,0.7) ellipse (2.5cm and 0.5cm);
	\draw[thick, dotted, color=blue] (1,2.2) ellipse (2.5cm and 0.5cm);
	
    \draw (0,0) -- (0,2) -- (2,2) -- (2,0) -- cycle;
    \draw (1.5,0.6) -- (1.5, 2.6) -- (3.5, 2.6) -- (3.5, 0.6) -- cycle;
    \draw (0,0) -- (1.5,0.6);
    \draw (0,2) -- (1.5, 2.6);
    \draw (2,0) -- (3.5,0.6);
    \draw (2,2) -- (3.5, 2.6);
	
	\end{tikzpicture}
	\caption{The deliberative space in Example~\ref{ex:3hyper}.}
	\label{fig:3hyper}
\end{figure}

\subsection{Multi-Party Compromises}

\begin{definition}[$\ell$-Compromise Transition]
\label{def:dcompromise}
Given an integer $\ell\ge 2$, we say that
a pair of coalition structures $(\D, \D')$ forms an {\em $\ell$-compromise transition}
if there exist coalitions $\d_1, \ldots, \d_\ell \in \D$, $\d'_1, \ldots,  \d'_\ell, \d' \in \D'$
and a proposal $p$ 
such that, for all $1 \leq i \leq \ell$ it holds that
$\D\setminus\{\d_1, \ldots, \d_\ell \} = \D' \setminus \{ \d'_1, \ldots, \d'_\ell, \d' \}$, 
$\d_i = (C_i, p_i)$,
$\d' = (\bigcup_{1 \leq i \leq \ell} C_i^p, p)$, 
$\d'_i = (C_i\setminus C_i^p, p_i)$,
and $|\bigcup_{1 \leq i \leq \ell} C_i^p|>|C_i|$.
\end{definition}
Accordingly, a compromise transition (Definition \ref{def:compromise}) is then a $2$-compromise 
transition. Note that our definition of an $\ell$-compromise does not require
that $C_i^p\neq\emptyset$ for all $i\in [\ell]$, i.e.,  
some coalitions $C_i$ may actually contribute an empty set $C_i^p$ of agents 
to the newly formed coalition (though for $\ell = 2$ this cannot happen). 
Hence, for every $\ell, \ell'$ with $2\le \ell<\ell'\le |\D|$,
it holds that an $\ell$-compromise transition $(\D, \D')$ is also an $\ell'$-compromise transition, as any coalition $C$ not involved in the $\ell$-compromise trivially contributes $\emptyset$ to the newly formed coalition. 

Recall now the total order $<_{\mathrm{lex}}$ defined in Equation~\eqref{eq:lexdef}. 
Observe that if $(\D, \D')$ is an $\ell$-compromise transition with $\ell\ge 2$ then 
$\gamma(\D)<_{\mathrm{lex}}\gamma(\D')$, because the newly formed coalition is strictly 
larger than all coalitions that have contributed to it. Consequently, all coalition structures
that arise in a deliberation that consists of $\ell$-compromise transitions are pairwise distinct,  
and hence the length of such deliberation is bounded by $n^n$.

\begin{corollary}\label{cor:lex}
For each $\ell\ge 2$, a deliberation that consists of $\ell$-compromise transitions can have at most $n^n$ transitions.
\end{corollary}

At the extreme, if the agents can engage in $n$-compromises, 
then a successful coalition structure can be reached in one step, by having all agents who approve a proposal $p\in M^*$ form a coalition around $p$. However, we are interested in outcomes that can be achieved when each interaction only involves a few coalitions. Thus, given a deliberation space, we are interested in finding the smallest value of $\ell$ such that $\ell$-compromise transitions can reach a successful outcome. Formally, let $\ell^*(d)$
be the smallest value of $\ell$ such that in a $d$-hypercube a deliberation that consists of $\ell$-compromise transitions
is necessarily successful.

We will now establish some lower and upper bounds on $\ell^*(d)$,
which hold for all $d\ge 3$. The proofs are rather technical, so they are 
relegated to Appendix~\ref{sec:proofs}. We note that, after the conference version
of our paper has been published, \citet{followup} 
strengthened the lower bound on $\ell^*(d)$ (Proposition~\ref{prop:hyper}) to $2^{\Theta(d)}$.

\begin{proposition} \label{prop:hyper}
For each $d\ge 3$ we have $\ell^*(d)\ge d$.
\end{proposition}

It is instructive to contrast Proposition~\ref{prop:hyper} with the earlier 
Theorem~\ref{thm:success}: While $d$-Euclidean spaces are `rich' enough for $2$-compromises 
to guarantee success in any dimension, the $d$-hypercubes are `sparse',  
so we need to facilitate interactions among at least $d$ coalitions 
to guarantee success.

For the upper bound, an easy observation
is that $\ell^*(d)\le 2^d$: this is because
in a terminal coalition structure
there cannot be two distinct coalitions supporting the same proposal.
We will now show how to strengthen this bound by (almost) a factor of two. 

\begin{proposition} \label{prop:hyper-upper}
For each $d\ge 2$ we have $\ell^*(d) \le 2^{d-1}+\frac{d+1}{2}$.
\end{proposition}

In what follows, we compute the exact value of $\ell^*(d)$ for $d=3$ and $d=4$. For $d=3$, 
the lower bound of Proposition~\ref{prop:hyper} turns out to be tight; however, for $d=4$ 
this is no longer the case.

\begin{proposition}\label{prop:34hyper}
We have $\ell^*(3)=3$ and $\ell^*(4)=5$.
\end{proposition}

While the lower bound shown by \citet{followup} and the upper bound 
of Proposition~\ref{prop:hyper-upper} are quite close, it remains an open
problem to fully close the gap between them.


\section{Conclusions and Future Work}\label{sec:concl}

We proposed a formal model of deliberation for agent populations forming coalitions around proposals in order to change the status quo. 
We identified several natural modes of coalition formation, capturing the many coalitional effects that deliberative processes could support. We studied sufficient conditions for them to succeed in identifying maximally approved proposals.
We intend our model as a foundation for the study of mechanisms and systems  allowing communities to self-govern.

\smallskip

To the best of our knowledge, ours is the first model of deliberation focusing on iterative processes of coalition formation. The model lends itself to several avenues for future research.

As we have mentioned earlier in the paper, several technical challenges left open by 
our work, such as upper and lower bounds on the length of a sequence of compromise 
transitions and the values of $\ell^*(d)$, have been addressed by the follow-up work of 
\citet{followup}. However, there are many conceptual questions that have not been answered 
yet.

In particular, one can consider natural extensions of the model, such as other transition types or other types of metric spaces. Specifically,
for each type of deliberation (as captured by allowable transitions),  
it is important to understand exactly how `rich' a space should be in order to support deliberation success. 
One can also ask whether our positive results are preserved if we impose additional conditions on the structure of the deliberative process,
e.g., require new proposals to be `close' to the original proposals.
We may also revisit our approach to modeling coalitions that support the status quo:
we now assume that each agent~$v$ with $X^v\neq\emptyset$
is capable of identifying some proposal in~$X^v$, and this assumption may be too
strong for many deliberation scenarios. It is perhaps more realistic to assume
that some agents start out by supporting the status quo, but then they learn about
a new proposal $p$ that they prefer to the status quo by observing a coalition
that supports $p$, and subsequently move to join this coalition. A formal model that can capture this form of preference discovery is a challenge for future work.

Further afield, it would be interesting
to consider a stochastic variant of our model, in which each transition from a state
is assigned a certain probability. 
A Markovian analysis of such systems 
might shed further light on the behavior of deliberative processes.

Crucially, in our current model agents are cooperative and non-strategic: they share a 
motivation to form larger coalitions and they truthfully reveal whether they approve a given 
proposal. Yet, agents may prefer to form a smaller group around proposals they strongly prefer 
to being in a larger coalition that supports a proposal they are lukewarm about. 
Furthermore, agents may be strategic about which proposals to support. 
This calls for exploring game-theoretic extensions of our 
framework. 

Yet another ambitious direction for future work is to design a practical tool for deliberation
and self-governance of an online community that incorporates the insights from our analysis. 
Such a tool could take the form of an AI bot over existing on-line deliberation platforms 
such as the LiquidFeedback and Polis platforms mentioned earlier in the paper.
The bot would suggest proposals to agents 
in order to engender compromise, hopefully fostering successful deliberations.


\appendix

\section{Auxiliary Notation and Proofs for Section~\ref{section:hypercube}}
\label{sec:proofs}

In what follows, given a proposal 
$p = (i_1, \dots, i_d)\in\{0, 1\}^d$ and a coalition structure $\D$, 
we denote the deliberative coalition in $\D$ that supports $p$ by $\d_p=(C_p, p)$; also, if an agent's ideal point is $p$, we will say that this agent is of {\em type} $p$. We will refer to the quantity $i_1+\dots+i_d$ as the {\em weight} of proposal $p$.

\begin{proof}[Proposition~\ref{prop:hyper}]
Fix some $d\ge 3$. We prove the claim by constructing a coalition structure in the $d$-hypercube that is terminal for $(d-1)$-compromise transitions, but not successful. The construction generalizes Example~\ref{ex:3hyper}. 

For each $i\in [d]$, let $e_i$ be the point in the hypercube that has $1$ in the $i$-th coordinate and $0$ in all other coordinates. Also, for each $i\in [d]\setminus\{1\}$, let $f_i$ be the point in the hypercube that has $1$ in coordinates $1$ and $i$, and $0$ in all other coordinates.

We create one agent of type $e_1$, and place her into a singleton coalition that supports 
$e_1$. For each $i=2, \dots, d$ we create $d-2$ agents of type $e_i$ and one agent of type 
$f_i$, and place them into a coalition that supports $e_i$; this coalition has $d-1$ members. 
Let $\D$ be the resulting coalition structure.

Note that agents of type $e_i$, $i\in [d]$, do not approve any proposals other than $e_i$. An 
agent of type $f_i$ for some $i\in [d]\setminus\{1\}$ approves $f_i$, $e_1$ and $e_i$. Hence, 
$e_1$ is the most-approved proposal, as it is approved by $d$ agents. However, there are no 
$\ell$-compromise transitions from $\D$ for $\ell <d$: unless each coalition 
in $\D$ contributes an agent to a coalition around $e_1$,  
the resulting coalition would have at most $d-1$ members, and hence it would not be 
attractive to agents of type $f_i$ with $i>1$.
\qed
\end{proof}

The proofs of Propositions~\ref{prop:hyper-upper} and~\ref{prop:34hyper} make
use of the two lemmas stated and proved below.

\begin{lemma}\label{lem:hyper}
In every deliberation space, every deliberation that consists
of $\ell$-compromise transitions and terminates in a coalition structure $\D$ with $|\D| = \ell$, is successful. 
\end{lemma}
\begin{proof}
 Consider a coalition $\D$ with $|\D|= \ell$ that is not successful. 
 Suppose first that $\D$ does not contain a coalition that supports $r$.
 Let $\d_1$, \ldots, $\d_\ell$ be the $\ell$ coalitions in $\D$. We
 have $|\d_i|<m^*$ for $i\in[\ell]$. For each $p \in M^*$ we have 
$V^p = \bigcup_{i\in[\ell]} \d_i^p$ and hence
 $|\bigcup_{i\in[\ell]} \d_i^p|=m^* >|\d_i|$. Thus, by Definition~\ref{def:dcompromise}, 
 there exists an $\ell$-compromise transition from $\D$ in which 
 agents in $\bigcup_{i\in [\ell]} \d_i^p$ form a coalition around $p$. 
If $\D$ contains a coalition that supports $r$, the same argument applies with $\ell-1$ coalitions. 
\qed
\end{proof}

\begin{lemma}\label{lem:uniquecoal}
Let $\D$ be a terminal coalition structure in the $d$-hypercube
with respect to $\ell$-compromises for some $d\ge 2$, $\ell\ge 2$.
Then for each point $p$ in the $d$-hypercube it holds that~$\D$ contains at most one deliberative coalition that contains
agents of type $p$.
\end{lemma}
\begin{proof}
Suppose $\D$ contains two deliberative coalitions $\d_1$, $\d_2$
that contain agents of type $p$. We can assume without
loss of generality that $|\d_1|\ge |\d_2|$. But then the agents of type $p$ approve the proposal supported by $\d_1$, so the agents
of type $p$ in $\d_2$ have an incentive to move to $\d_1$,
a contradiction with $\D$ being terminal.
\qed
\end{proof}

\begin{proof}[Proposition~\ref{prop:hyper-upper}]
The claim is clearly true for $d=2$, so we can assume $d\ge 3$.
Consider a deliberative coalition structure $\D$ that does not admit any $\ell$-transitions 
for $\ell=2^{d-1}+\frac{d+1}{2}$. We say that 
a coalition $\d=(C, p)$ in $\D$ is {\em pure} if all agents 
in $C$ have the same type. We claim that $\D$ can contain
at most $d+1$ pure coalitions. 

Indeed, suppose that $\D$ contains $d+2$ pure coalitions.
By the pigeonhole principle, there are two pure coalitions
$\d=(C, p)$ and $\d'=(C', p')$ in $\D$ such that all agents in $\d$ are of type $i_1i_2\dots i_d$, all agents in $\d'$ are of type $j_1j_2\dots j_d$,
and there exists an $s\in [d]$ such that $i_s=j_s=1$. 
Let $e_s$ be a point in the hypercube that has $1$ in position $s$ and $0$ in all other positions. Then
there is a transition in which all agents in $C$, all agents in $C'$ and all agents currently supporting $s$ form a coalition that supports $s$; this is a contradiction, as this transition only involves $3$ coalitions, and $2^{d-1}+\frac{d+1}{2}>3$ for $d\ge 3$.

Thus, all but $d+1$ coalitions in $\D$ are not pure, 
and hence contain agents of at least two different types. 
Further, by Lemma~\ref{lem:uniquecoal}, for each $p\in \{0, 1\}^d$ 
it cannot be the case that two coalitions in $\D$ contain agents of type $p$.
Hence, if $\D$ contains $t$ pure coalitions, it
contains at most $(2^{d}-t)/2$ non-pure coalitions, 
i.e., at most $2^{d-1}+t/2$ coalitions altogether;
as $t\le d+1$, our bound on $\ell^*(d)$ now follows from Lemma~\ref{lem:hyper}.
\qed
\end{proof}

\begin{proof}[Proposition~\ref{prop:34hyper}]
We partition the proof into four claims.
\begin{claim}
$\ell^*(3)\ge 3$.
\end{claim}
\begin{proof}
This follows immediately from Proposition~\ref{prop:hyper}.
\end{proof}
\begin{claim}
$\ell^*(3)\le 3$.
\end{claim}
\begin{proof}
Let $\D$ be a terminal coalition structure, and suppose that $\D$ is not successful.

Let ${\mathcal C}=\{C_{110}, C_{101}, C_{011}, C_{111}\}$.
Note that all agents in $C_{110}$ are of type $110$ or $111$, 
and similarly for $C_{101}$ and $C_{011}$. On the other hand, $C_{111}$ may contain agents of type $111$, $110$, $101$ or $011$.

Suppose first that at least two of the coalitions 
in $\mathcal C$ are not empty.
Then these two coalitions have a merge transition (which is also a $3$-compromise transition) in which they form a coalition around $111$, a contradiction with~$\D$ being a terminal coalition structure.
On the other hand, if all coalitions in $\mathcal C$ are empty, then $|\D|=3$, and we obtain a contradiction by Lemma~\ref{lem:hyper}.

Thus, we can focus on the case where exactly one coalition in $\mathcal C$---say, $C$---is non-empty.
Lemma~\ref{lem:hyper} then implies that each of the coalitions $C_{001}$, $C_{010}$, and~$C_{100}$ is non-empty.
As argued above, if there is an agent in $C$ of type $(i, j, k)$ then $i+j+k\ge 2$.
Now, if $C$ contains no agent of type $110$, 
then all agents in $C$ approve $001$ and hence there is a follow transition where $C$ joins $C_{001}$, a contradiction with $\D$ being a terminal coalition. Similarly, if $C$ contains no agent of type $101$, then $C$ can join $C_{010}$,
and if $C$ contains no agent of type $011$,
then $C$ can join $C_{100}$. We conclude that $C$ contains 
agents of types $110$, $101$, and $011$; hence, it has to be the case that $C$ supports $111$, 
i.e., $C=C_{111}$.

Now, suppose that $|C_{111}|\le |C_{ijk}|$ for some $i, j, k$ with $i+j+k=1$; assume without loss of generality that $|C_{111}|\le |C_{001}|$. We have already argued that $C_{111}$
contains some agents of type $011$;
note that these agents approve $001$.
Hence there is a subsume transition in which all agents in $C_{111}$ who approve $001$ move to $001$,
a contradiction with $\D$ being a terminal coalition structure.

The remaining possibility is that
$|C_{111}|> \max\{|C_{001}|, |C_{010}|, |C_{100}|\}$. 
In this case, if for some $i, j, k$ with $i+j+k=1$ 
the coalition $C_{ijk}$ contained some agents who approved $111$, then there would be a subsume transition in which all these agents would move to $111$.
Thus $C_{111}$ already contains all agents who approve $111$.
As we assume that $\D$ is not successful, we have $|C_{111}|<m^*$ and hence $111\not\in M^*$. It follows that $110$,
$101$, and $011$ are not in $M^*$ either, as every agent who approves one of these proposals also approves $111$.
Moreover, as $C_{111}$ contains all agents that approve $111$, 
each coalition $C_{ijk}$ with $i+j+k=1$ consists exclusively of agents of type $(i, j, k)$. 
Thus, if, say, $100$ is in $M^*$, then all agents who approve $100$ are either 
in $C_{100}$ or in $C_{111}$, and hence there is a follow transition in which all these agents join forces around $100$, forming a coalition of size $m^*$, a contradiction with $\D$ being terminal.
\end{proof}

\begin{claim}
$\ell^*(4)\ge 5$.
\end{claim}
\begin{proof}
Consider a coalition structure 
$\D = \{\d_{1000}, \d_{0100}, \d_{0010}, \d_{0001}, \d_{1110}\}$,
where 
\begin{itemize}
\item 
$C_{1000}$ contains $21$ agents of type $1000$ and $10$ agents of type $1001$;
\item
$C_{0100}$ contains $21$ agents of type $0100$ and $10$ agents of type $0101$;
\item
$C_{0010}$ contains $21$ agents of type $0010$ and $10$ agents of type $0011$;
\item
$C_{0001}$ contains one agent of type $0001$; and
\item
$C_{1110}$ contains
$30$ agents of type $1101$,
$30$ agents of type $1011$,
$30$ agents of type $0111$,
$10$ agents of type $1100$,
$10$ agents of type $1010$, and
$10$ agents of type $0110$.
\end{itemize}
One can verify that each agent approves the proposal supported
by her deliberative coalition.

There are $121$ agents who approve $0001$: this includes the one
agent in $C_{0001}$, $10$ agents from each of the coalitions
$C_{1000}$, $C_{0100}$ and $C_{0001}$, and $90$ agents from $C_{1110}$. 
For these agents to gather at $0001$, all five coalitions need to be involved. In particular, it cannot be the case that only the agents of types $1001$, $0101$ and $0011$ move to $0001$, as this would result in a coalition of size 
$31$, which is equal to the size of their current coalitions.

It remains to argue that $\D$ is stable with respect to transitions that involve at most four coalitions. The argument in the previous paragraph already shows that every such transition should involve agents in 
$C_{1110}$. 

For each $s=1, 2, 3, 4$ we will argue that there is no transition
in which agents move to a point of weight $s$.

\begin{itemize}
\item
$s=4$:
It is immediate that there is no transition to the unique 
point of weight $4$, i.e., $1111$, as all agents who approve $1111$
are currently in $C_{1110}$.
\item
$s=3$:
No agent not currently in $C_{1110}$ approves $1110$,
so there is no transition in which some agents move to $1110$.
Now, consider a point of weight $3$ that differs from $1110$;
for concreteness, take $0111$ (the remaining two cases can be analyzed in the same way). 
There are $100$ agents 
in $C_{1110}$ that approve $0111$, as well as $10$ agents in $C_{0100}$
and $10$ agents in $C_{0011}$. Thus, the coalition at $0111$
would contain $120$ agents, which is equal to the size of $C_{1110}$.
Therefore, agents in $C_{1110}$ would not benefit from this transition.
\item
$s=2$:
For every proposal of weight $2$, there are at most $60$ 
agents in $C_{1110}$ who approve this proposal and 
at most $10$ other agents who approve it; as $|C_{1110}|>70$,
there is no transition to a point of weight $2$ that
is attractive to agents in $C_{1110}$.
\item
$s=1$:
There are $80$ agents in $C_{1110}$ who approve $1000$;
these agents could move to $1000$, but this would result 
in a coalition of size $111$, which is smaller than 
their current coalition. Hence, a transition in which agents
who approve $1000$ move to $1000$ is not possible; for a similar
reason, there is no transition where agents who approve $0100$
or $0010$ move to the respective points.
\end{itemize}
\end{proof}

\bigskip

\begin{claim}
$\ell^*(4)\le 5$.
\end{claim}
\begin{proof}
Let $\D$ be a terminal coalition structure with respect to $5$-compromise transitions, 
and suppose for the sake of contradiction that $\D$ is not successful.
By Lemma~\ref{lem:hyper}, $\D$ contains at least six coalitions.

Consider an agent who approves $1111$. 
Then her ideal point has weight $3$ or $4$. Consequently, 
this agent approves all proposals of weight $3$ or $4$: indeed,
her distance to $0000$ is at least $3$, while her distance to any proposal
of weight at least $3$ is at most $2$. It follows that if $\D$
contains a coalition $\d_{1111}= (C_{1111}, 1111)$ then it does
not contain any coalitions that support a proposal of weight $3$:
indeed, the ideal point of each agent in $C_{1111}$ has weight $3$
or $4$, so if there is another coalition $C'$ that supports a proposal
$p$ of weight $3$, then there is a follow transition in which $C_{1111}$
and $C'$ merge around $p$. 

Furthermore, $\D$ can contain at most two coalitions that support proposals
of weight at least~$3$. Indeed, assume for the sake of contradiction that there are three
coalitions in $\D$ that support proposals of weight at least $3$. As argued
in the previous paragraph, it cannot be the case that one of them supports $1111$, so we can assume without loss of generality that 
$\D$ contains coalitions 
$\d_{1110}= (C_{1110}, 1110)$, 
$\d_{1101}= (C_{1101}, 1101)$, and 
$\d_{1011}= (C_{1011}, 1011)$, 
with $|C_{1110}|\ge |C_{1101}|\ge |C_{1011}|$. Then $C_{1011}$ can only contain agents of type $0011$: 
all other agents who approve $1011$ also
approve either $1110$ or $1101$, so they would have
an incentive to move to the weakly larger coalitions that support these proposals.
Similarly, $C_{1101}$ can only contain agents of type
$1001$ or $0101$, as all other agents
who approve $1101$ also approve $1110$, so they can
move to the weakly larger coalition $C_{1110}$.
But then all agents in $C_{1101}$ and $C_{1011}$ approve $0001$,
so there is a $3$-compromise in which all agents in $C_{1101}$,
$C_{1011}$ and $C_{0001}$ merge around $0001$, a contradiction.

Thus, we have argued that $\D$ contains at most two coalitions that support
proposals of weight $3$ or higher. Now, consider coalitions in $\D$ that support proposals of weight at most~$2$.

Suppose $\D$ contains a coalition $\d_{1100}= (C_{1100}, 1100)$.
Note that all agents in $C_{1100}$ approve $1000$ and $0100$.
Thus, if $\D$ contains a coalition that supports $1000$ 
or $0100$, then there would be a follow transition in which
all agents in $C_{1100}$ adopt the proposal of this coalition, a contradiction with $\D$ being a terminal coalition structure. Similarly, if $\D$ contains a coalition $\d_{1010}$ or $\d_{1001}$, 
there would be a merge transition in which
$C_{1100}$ merges with this coalition around $1000$, and if $\D$ contains a coalition $\d_{0110}$ or $\d_{0101}$, 
there would be a merge transition in which
$C_{1100}$ merges with this coalition around $0100$.

We conclude that if $\D$ contains $\d_{1100}$ then it contains
at most two other coalitions that support proposals of weight at most 
$2$: namely, it may contain $\d_{0011}$ (in which case it contains
no coalitions that support a proposal of weight $1$), or it may contain
one or both of the coalitions $\d_{0010}$ and $\d_{0001}$.
By symmetry, it follows that if $\D$ contains a coalition that 
supports a proposal of weight $2$, then it contains at most
three coalitions that support proposals of weight at most $2$.
As we have argued that $\D$ contains at most two proposals of weight 
at least $3$, we obtain a contradiction with $|\D|\ge 6$.
Hence, we may assume that $\D$ contains no coalitions that support
a proposal of weight exactly $2$.

Thus, if $|\D|\ge 6$, it must be the case that $\D$ contains four coalitions supporting proposals of weight $1$, as well as two coalitions supporting
proposals of weight at least $3$. We have argued that among these two coalitions, one (a weakly smaller one) only contains agents of two types
of weight $2$ each; moreover, our analysis shows that all agents in this coalition approve some proposal $p$ of weight $1$. But then there is a follow transition in which this coalition joins the coalition that currently
supports $p$, a contradiction again. Thus, we can conclude that $\D$ is successful.
\end{proof}
Proposition~\ref{prop:34hyper} now follows by combining Claims 1--4.
\qed
\end{proof}


\bibliographystyle{plainnat}

\begin{thebibliography}{32}
\providecommand{\natexlab}[1]{#1}
\providecommand{\url}[1]{\texttt{#1}}
\expandafter\ifx\csname urlstyle\endcsname\relax
  \providecommand{\doi}[1]{doi: #1}\else
  \providecommand{\doi}{doi: \begingroup \urlstyle{rm}\Url}\fi

\bibitem[Austen-Smith and Feddersen(2005)]{austen-smith05deliberation}
David Austen-Smith and Timothy~J. Feddersen.
\newblock Deliberation and voting rules.
\newblock In \emph{Social Choice and Strategic Decisions}, Studies in Social
  Choice and Welfare. Springer, 2005.

\bibitem[Aziz and Savani(2016)]{aziz16handbook}
Haris Aziz and Rahul Savani.
\newblock Hedonic games.
\newblock In Felix Brandt, Vincent Conitzer, Ulle Endriss, J{\'{e}}r{\^{o}}me
  Lang, and Ariel Procaccia, editors, \emph{Handbook of Computational Social
  Choice}, chapter~15. Cambridge University Press, 2016.

\bibitem[Behrens et~al.(2014)Behrens, Kistner, Nitsche, and
  Swierczek]{liquid_feedback}
Jan Behrens, Axel Kistner, Andreas Nitsche, and Bj\"orn Swierczek.
\newblock \emph{Principles of Liquid Feedback}.
\newblock Interaktive Demokratie, 2014.

\bibitem[Brill(2018)]{brill18interactive}
Markus Brill.
\newblock Interactive democracy.
\newblock In \emph{Proceedings of AAMAS '18}, pages 1183--1187, 2018.

\bibitem[Chalkiadakis and Boutilier(2008)]{chalkiadakis08sequential}
Georgios Chalkiadakis and Craig Boutilier.
\newblock Sequential decision making in repeated coalition formation under
  uncertainty.
\newblock In \emph{Proceedings of AAMAS '08}, pages 347--354, 2008.

\bibitem[Chung and Duggan(2020)]{chung20formal}
Hun Chung and John Duggan.
\newblock A formal theory of democratic deliberation.
\newblock \emph{American Political Science Review}, 114\penalty0 (1):\penalty0
  14--35, 2020.

\bibitem[Coombs(1964)]{coombs}
Clyde~H. Coombs.
\newblock \emph{A theory of data}.
\newblock Wiley, 1964.

\bibitem[Darmann et~al.(2018)Darmann, Elkind, Kurz, Lang, Schauer, and
  Woeginger]{gasp}
Andreas Darmann, Edith Elkind, Sascha Kurz, J{\'e}r{\^o}me Lang, Joachim
  Schauer, and Gerhard Woeginger.
\newblock Group activity selection problem with approval preferences.
\newblock \emph{International Journal of Game Theory}, 47:\penalty0 767--796,
  2018.

\bibitem[{de Vries}(1999)]{devries99governing}
Miranda W.~M. {de Vries}.
\newblock \emph{Governing with your closest neighbour: an assessment of spatial
  coalition formation theories}.
\newblock Radboud University Nijmegen, 1999.

\bibitem[Dieckmann and Schwalbe(2002)]{arnold02dynamic}
Tone Dieckmann and Ulrich Schwalbe.
\newblock Dynamic coalition formation and the core.
\newblock \emph{Journal of Economic Behavior and Organization}, 49\penalty0
  (3):\penalty0 363--380, 2002.

\bibitem[Elkind et~al.(2021)Elkind, Grossi, Shapiro, and Talmon]{preliminary}
Edith Elkind, Davide Grossi, Ehud Shapiro, and Nimrod Talmon.
\newblock United for change: Deliberative coalition formation to change the
  status quo.
\newblock In \emph{Proceedings of AAAI '21}, 2021.

\bibitem[Elkind et~al.(2022)Elkind, Ghosh, and Goldberg]{followup}
Edith Elkind, Abheek Ghosh, and Paul Goldberg.
\newblock Complexity of deliberative coalition formation.
\newblock In \emph{Proceedings of AAAI '22'}, pages 4975--4982, 2022.

\bibitem[Fain et~al.(2017)Fain, Goel, Munagala, and
  Sakshuwong]{fain2017sequential}
Brandon Fain, Ashish Goel, Kamesh Munagala, and Sukolsak Sakshuwong.
\newblock Sequential deliberation for social choice.
\newblock In \emph{Proceedings of WINE '17}, pages 177--190, 2017.

\bibitem[Glazer and Rubinstein(2001)]{glazer01debates}
Jacob Glazer and Ariel Rubinstein.
\newblock Debates and decisions: On a rationale of argumentation rules.
\newblock \emph{Games and Economic Behavior}, 36\penalty0 (2):\penalty0
  158--173, 2001.

\bibitem[Glazer and Rubinstein(2004)]{glazer04optimal}
Jacob Glazer and Ariel Rubinstein.
\newblock On optimal rules of persuasion.
\newblock \emph{Econometrica}, pages 1715--1736, 2004.

\bibitem[Glazer and Rubinstein(2006)]{glazer06study}
Jacob Glazer and Ariel Rubinstein.
\newblock A study in the pragmatics of persuasion: A game theoretical approach.
\newblock \emph{Theoretical Economics}, 1:\penalty0 395--410, 2006.

\bibitem[Goel and Lee(2016)]{goel2016towards}
Ashish Goel and David~T Lee.
\newblock Towards large-scale deliberative decision-making: Small groups and
  the importance of triads.
\newblock In \emph{Proceedings of EC '16}, pages 287--303, 2016.

\bibitem[Groot(1974)]{degroot74reaching}
Morris~De Groot.
\newblock Reaching a consensus.
\newblock \emph{Journal of the American Statistical Association}, 69\penalty0
  (345):\penalty0 118--121, 1974.

\bibitem[Hafer and Landa(2007)]{hafer07deliberation}
Catherine Hafer and Dimitri Landa.
\newblock Deliberation as self-discovery and institutions for political speech.
\newblock \emph{Journal of Theoretical Politics}, 19\penalty0 (3):\penalty0
  329--360, 2007.

\bibitem[Hotelling(1929)]{hotelling29stability}
Harold Hotelling.
\newblock Stability in competition.
\newblock \emph{Economic Journal}, 39\penalty0 (153):\penalty0 41--57, 1929.

\bibitem[Karanikolas et~al.(2019)Karanikolas, Bisquert, and
  Kaklamanis]{karanikolas2019voting}
Nikos Karanikolas, Pierre Bisquert, and Christos Kaklamanis.
\newblock A voting argumentation framework: Considering the reasoning behind
  preferences.
\newblock In \emph{Proceedings of ICAART '19}, pages 42--53, 2019.

\bibitem[Landa and Meirowitz(2009)]{landa09game}
Dimitri Landa and Adam Meirowitz.
\newblock Game theory, information, and deliberative democracy.
\newblock \emph{American Journal of Political Science}, 53\penalty0
  (2):\penalty0 427--444, 2009.

\bibitem[Lang and Xia(2016)]{ma-comsoc}
J{\'{e}}r{\^{o}}me Lang and Lirong Xia.
\newblock Voting in combinatorial domains.
\newblock In Felix Brandt, Vincent Conitzer, Ulle Endriss, J{\'{e}}r{\^{o}}me
  Lang, and Ariel Procaccia, editors, \emph{Handbook of Computational Social
  Choice}, chapter~9, pages 197--222. Cambridge University Press, 2016.

\bibitem[List(2011)]{list11group}
Christian List.
\newblock Group communication and the transformation of judgments: An
  impossibility result.
\newblock \emph{The Journal of Political Philosophy}, 19\penalty0 (1):\penalty0
  1--27, 2011.

\bibitem[List et~al.(2013)List, Luskin, Fishkin, and
  McLean]{list13deliberation}
Christian List, Robert Luskin, James Fishkin, and Iain McLean.
\newblock Deliberation, single-peakedness, and the possibility of meaningful
  democracy: Evidence from deliberative polls.
\newblock \emph{The Journal of Politics}, 75\penalty0 (1):\penalty0 80--95,
  2013.

\bibitem[Merrill and Grofman(1999)]{MG99}
Samuel Merrill and Bernard Grofman.
\newblock \emph{A unified theory of voting: Directional and proximity spatial
  models}.
\newblock Cambridge University Press, 1999.

\bibitem[Patty(2008)]{patty08arguments}
John~W. Patty.
\newblock Arguments-based collective choice.
\newblock \emph{Journal of Theoretical Politics}, 20\penalty0 (4):\penalty0
  379--414, 2008.

\bibitem[Perote-Pe{\~n}a and Piggins(2015)]{perote15model}
Juan Perote-Pe{\~n}a and Ashley Piggins.
\newblock A model of deliberative and aggregative democracy.
\newblock \emph{Economics and Philosophy}, 31\penalty0 (1):\penalty0 93--121,
  2015.

\bibitem[Roy and Rafiee~Rad(2020)]{roy20deliberation}
Olivier Roy and Soroush Rafiee~Rad.
\newblock Deliberation, single-peakedness, and coherent aggregation.
\newblock \emph{American Political Science Review}, 2020.

\bibitem[Shahaf et~al.(2019)Shahaf, Shapiro, and Talmon]{metricspaces}
Gal Shahaf, Ehud Shapiro, and Nimrod Talmon.
\newblock Aggregation over metric spaces: Proposing and voting in elections,
  budgeting, and legislation.
\newblock In \emph{Proceedings of ADT '19}, 2019.

\bibitem[Speroni~di Fenizio and Velikanov(2016)]{pietro16system}
Pietro Speroni~di Fenizio and Cyril Velikanov.
\newblock System-generated requests for rewriting proposals.
\newblock \emph{CoRR}, abs/1611.10095, 2016.
\newblock URL \url{http://arxiv.org/abs/1611.10095}.

\bibitem[Zwicker(2016)]{comsoc-book}
William~S. Zwicker.
\newblock Introduction to the theory of voting.
\newblock In Felix Brandt, Vincent Conitzer, Ulle Endriss, J{\'{e}}r{\^{o}}me
  Lang, and Ariel Procaccia, editors, \emph{Handbook of Computational Social
  Choice}, chapter~2, pages 23--56. Cambridge University Press, 2016.

\end{thebibliography}

\end{document}